\documentclass[smallextended]{svjour3}       
\smartqed  
\usepackage{appendix}
\usepackage{amsmath}
\usepackage{graphicx}
\usepackage{lineno}
\usepackage{array}
\usepackage{longtable}
\usepackage{natbib}
\usepackage{multirow}

\usepackage{amsmath,amssymb,amsfonts,amsthm}
\usepackage{subfigure,epstopdf}
\usepackage{hyperref}

\usepackage{xcolor}
\definecolor{boxgrey}{HTML}{F3F3F3}

\newcommand{\hlbox}[2]{%
  \begin{center}%
    \fcolorbox{white}{boxgrey}{%
      \parbox{.9\columnwidth}{\noindent \textbf{#1}. \textit{#2}}
    }%
  \end{center}%
}

\newtheorem{prop}{Proposition}
\newcommand*\patchAmsMathEnvironmentForLineno[1]{%
\expandafter\let\csname old#1\expandafter\endcsname\csname #1\endcsname
\expandafter\let\csname oldend#1\expandafter\endcsname\csname end#1\endcsname
\renewenvironment{#1}%
{\linenomath\csname old#1\endcsname}%
{\csname oldend#1\endcsname\endlinenomath}}%
\newcommand*\patchBothAmsMathEnvironmentsForLineno[1]{%
\patchAmsMathEnvironmentForLineno{#1}%
\patchAmsMathEnvironmentForLineno{#1*}}%
\AtBeginDocument{%
\patchBothAmsMathEnvironmentsForLineno{equation}%
\patchBothAmsMathEnvironmentsForLineno{align}%
\patchBothAmsMathEnvironmentsForLineno{flalign}%
\patchBothAmsMathEnvironmentsForLineno{alignat}%
\patchBothAmsMathEnvironmentsForLineno{gather}%
\patchBothAmsMathEnvironmentsForLineno{multline}%
}

\begin{document}

\title{Robust Reputation Independence in Ranking Systems for Multiple Sensitive Attributes}

\titlerunning{Robust Reputation Independence for Multiple Sensitive Attributes}        

\author{Guilherme Ramos \and Ludovico Boratto \and Mirko Marras}
\authorrunning{Ramos et al.}

\institute{
           G. Ramos \at University of Porto, Porto, Portugal \email{guilhermeramos21@gmail.com}           
           \and
           L. Boratto \at University of Cagliari, Cagliari, Italy 
           \email{ludovico.boratto@acm.org}     
           \and
           M. Marras \at \'EPFL, Lausanne, Switzerland 
           \email{mirko.marras@acm.org}            
}

\date{Received: 01 March 2021}

\maketitle

\begin{abstract}
Ranking systems have an unprecedented influence on how and what information people access, and their impact on our society is being analyzed from different perspectives, such as users' discrimination.
A notable example is represented by reputation-based ranking systems, a class of systems that rely on users' reputation to generate a non-personalized item-ranking, proved to be biased against certain demographic classes.
To safeguard that a given sensitive user's attribute does not systematically affect the reputation of that user, prior work has operationalized a reputation independence constraint on this class of systems. 
In this paper, we uncover that guaranteeing reputation independence for a single sensitive attribute is not enough. 
When mitigating biases based on one sensitive attribute (e.g., gender), the final ranking might still be biased against certain demographic groups formed based on another attribute (e.g., age). 
Hence, we propose a novel approach to introduce reputation independence for multiple sensitive attributes simultaneously. 
We then analyze the extent to which our approach impacts on discrimination and other important properties of the ranking system, such as its quality and robustness against attacks. 
Experiments on two real-world datasets show that our approach leads to less biased rankings with respect to multiple users' sensitive attributes, without affecting the system's quality and robustness. 

\keywords{Bias Mitigation \and Ranking Systems \and Reputation \and Robustness}

\end{abstract}

\section{Introduction}

Ranking systems are becoming a vital tool to access information on the Web, from search engines to recommender systems~\citep{gama2020new}. 
Given their role in our experience online, the results they produce must not harm users in any way. 
However, it is known that a biased ranking can lead to a loss of trust in the system \citep{PanHJLGG07}, and that 
a ranking can also hide a discrimination against users belonging to legally-protected classes~\citep{EkstrandBD19,DBLP:journals/umuai/BorattoUMUI21,DBLP:journals/umuai/BorattoIJAIED21}. 

Most of the effort in the literature on non-personalized rankings has been devoted on biases associated to the sensitive attributes of users being ranked ({\em user rankings} such as those ranking job candidates), by showing that minorities are under-exposed in a ranking. The impact of users' sensitive attributes in non-personalized {\em item rankings} is therefore still under-explored. 
When items are ranked, only the possible consequences for the item providers have been considered~\citep{MehrotraMBL018}. However, if item rankings are learned from user preferences and this learning is biased on users' sensitive attributes, those belonging to minority groups might be considered as less relevant by the system. 
\cite{RamosB20} studied the impact of demographic attributes on reputation-based ranking systems, a class of systems that ranks the items by weighting the users' ratings by their reputation in the system.
The reputation is computed by comparing a user's preference with that of the other users~\citep{medo2010effect,LiYHC12,saude2017robust}. 
Computing reputation by analyzing the user's conformance with respect to the community can be beneficial not only to provide rankings that better reflect the preferences of the community, but also when reputation is exploited to avoid attacks to the ratings, such as {\em bribing}~\citep{saude2017robust}. Indeed, anomalous ratings would be provided by a user with a low conformance (and, consequently, a low reputation), thus avoiding negative consequences for the platform. 
In case users of certain legally-protected groups systematically receive lower/higher reputation scores, the produced rankings would not reflect the community's preferences as a whole, and the rankings get polarized towards certain groups. 
This effect can have negative consequences on the end-users and on the platform, which might lose the trust of those unjustly recognized as having a low reputation. 
\cite{RamosB20} showed that the reputation scores are usually lower for minority demographic groups, and mitigated this bias by ensuring that reputation scores for users belonging to different legally-protected groups are statistically indistinguishable ({\em reputation independence}). 
However, they operate on a sensitive attribute individually and, as shown by~\cite{KleinbergMR17}, acting on groups characterized by a sensitive attribute does not necessarily provide guarantees to groups based on another sensitive attribute.  

Driven by this motivation, this paper investigates whether reputation independence on a single sensitive attribute (e.g., gender) provides, by extension, reputation independence when considering another sensitive attribute (e.g., age). 
We theoretically and experimentally show that this is not the case. 
Current definitions of reputation independence urge to be strengthened to preserve the rankings from biased notions of reputation on multiple sensitive attributes. 
For this reason, we introduce the novel concept of {\em reputation independence for multiple sensitive attributes} simultaneously. 
Given a set of sensitive attributes of the users and the classes each attribute can take, this concept guarantees that the reputation is independent from {\em all} these sensitive attributes. 

Besides ensuring that users do not interact with rankings that biasedly reflect the community's preferences, reputation-based ranking systems are often adopted to protect users from perspectives that go beyond ethical and societal aspects. One of the most important is for the system to be {\em robust} against attacks that might alter the ranking, such as bribing.
In bribing attacks, an entity gives incentives to users to change their or add ratings so that an item can increase its exposure in a ranking. 
Though a non-robust ranking algorithm has been proved to produce biased results because the items' ratings introduced through bribing lead to an essentially different ranking~\citep{RamosBC20}, no work has ever considered the opposite perspective and studied {\em if debiasing procedures might affect the ranking systems' robustness}. 
Therefore, in this paper, we also assess the robustness of the treated rankings under several attack types, to investigate whether protecting users from an ethical and societal perspective (e.g., reputation independence) can expose them to security risks.
Our results show that it is possible to protect the users by securing reputation independence from multiple sensitive attributes without affecting the system's robustness. 
Hence, a platform employing the proposed approach to shape their rankings can offer guarantees of less biased results and robustness for the users, increasing their trust to the platform. This, in pair with offering effective and accurate results, has clear benefits on the overall business of a platform. Concretely, we will show that it is possible to de-bias the reputation scores from the sensitive attributes of the users, thus producing rankings that better reflect the individual preferences, without affecting neither the robustness of the platform nor the ranking quality. Hence, by touching a beyond-accuracy perspective of the rankings (reputation  independence), we do not affect primary properties of the service offered to the user (robustness and ranking quality).

Specifically, our contributions can be summarized as follows:
\begin{itemize}
    \item We provide evidence, both theoretically and experimentally, that, for reputation independence to be guaranteed, it should cover multiple sensitive attributes of the users.
    \item We extend the existing notion of reputation independence and operationalize it, to embrace more than one sensitive attribute simultaneously, and study its complexity.
    \item We assess the extent to which our approach creates rankings based on less biased reputations and compare it against state-of-the-art reputation debiasing solutions; our approach ensures multi-attribute reputation independence. 
    \item We assess the robustness and the ranking quality achieved by our approach and compare it against state-of-the-art systems; while ensuring multi-attribute reputation independence, our approach does not sacrifice robustness and ranking quality.
\end{itemize}

The rest of the paper is structured as follows. 
Section~\ref{sec:related} presents related work and Section~\ref{sec:notation} introduces preliminary concepts. 
Then, Section~\ref{sec:single} analyzes the extent to which state-of-the-art reputation debiasing can produce unbiased rankings for different sensitive attributes. 
We present the concept of multi-attribute reputation independence in Section~\ref{sec:mit_mul} and study its impact on bias and robustness in Section~\ref{sec:evaluation}. 
Section~\ref{sec:impact} discusses how our approach is beneficial for online ecosystems, and Section~\ref{sec:conc} finally paints lines of future research. 

\section{Related work}\label{sec:related}
This section covers related work. 
We start by analyzing the impact of biased non-personalized rankings. 
Subsequently, we present literature regarding robustness in reputation-based ranking systems. 

\paragraph{Impact of biased non-personalized rankings.} 
Ranking composition strongly affects our perception of the quality of the ranked items. 
The eye-tracking study presented in~\cite{PanHJLGG07} showed that users trust the order in which the systems rank items, deeming as more relevant those on the top of the ranking. 
Nevertheless, most search engines exhibit bias in the ranking, which the users are not aware of~\citep{KulshresthaEMZG19}. 
Bias can emerge in multiple ways, such as during the learning phase of query-based ranking algorithms that analyze the items selected by the users via their clicks; to overcome this bias, unbiased learning-to-rank approaches have been proposed~\citep{joachims2017unbiased}.
These biases might take several forms, such as polarization towards political parties or providing more exposure to specific individuals according to their sensitive attributes~\citep{EkstrandBD19}. 
This last form of bias is currently receiving much attention since the bias associated with sensitive attributes might lead to undesired phenomena. 
For example, the system can end up discriminating against individuals belonging to minority groups of a legally-protected class~\citep{HajianBC16}. 
In non-personalized systems, these effects were mostly analyzed in people rankings. 
In~\cite{DiazMEBC20,SinghJ18,Zehlike020}, the authors studied the exposure given to individuals. 
Additionally, \cite{ZehlikeB0HMB17} provides guarantees that a group is present for a certain proportion in the rankings. 
Usually, the amount of exposure or visibility a group should receive is based on its representation in the data~\citep{BiegaGW18,SapiezynskiZRMW19,SinghJ18,YadavDJ19}. 
In the context of personalized rankings, such as recommender systems that generate different suggestions for each user, the impact of sensitives attributes on ranking quality (fairness) is also considered. 
This topic goes beyond the scope of this paper. 
So, we remind the reader of the recent survey by \cite{AbdollahpouriAB20} for references in this area. 
Compared with prior work, our paper focuses on a class of non-personalized item-rankings that rely on users' reputations to generate ranking, namely reputation-based ranking systems. 
The extent to which bias against certain groups of individuals, characterized by a common sensitive attribute, affects the reputations leveraged by this class of systems is still under-explored. 
Indeed, contributions tackling bias against groups associated with  \emph{individual sensitive attributes} have been recently proposed \citep{RamosB20}. 

\paragraph{Robustness in reputation-based ranking systems.} 
Prior reputation-based ranking systems have employed a weighted average as a strategy to combine individual ratings. \cite{yu2006decoding,de2010iterative} proposed relevant examples of works in this direction. \cite{LiYHC12} introduced the concept of {\em reputation}, which measures how close are the preferences of a user to those of the others. However, reputation can also be computed considering other data sources, such as product categories~\citep{li2015topic}, or by considering notions of trust~\citep{AllahbakhshIMB15}. One essential property of this class of systems is represented by their robustness, which was studied from different perspectives. \cite{rezvani2014secure} aims at improving robustness against collusion attacks by providing an approximation of the existing iterative filtering techniques, while~\cite{SU201755,xu2019meurep} seek to provide robustness when considering quality-of-service data. 
The approach proposed by \cite{tibermacine2019reputation} is a HITS-based reputation evaluation process that allows us to detect malicious users based on a majority voting and assess service reputation after the exclusion of malicious users' feedback ratings.

\paragraph{Contextualizing our contribution.} 
To the best of our knowledge, no prior work has ever considered the combination of the impact of multiple users' sensitive attributes to generate a less biased ranking and the possible impact of a bias mitigation strategy on the system's robustness. Given that both perspectives equally impact the end-users and the platform, in the rest of the paper, we investigate techniques for removing bias in the reputation scores computed by a ranking system and how they affect robustness. 

\section{Preliminaries}\label{sec:notation}
We formalize the main concepts underlying our study, including the ranking context and the reputation-based ranking systems. 
We close this section by presenting the datasets considered in our study, which will serve as a means to assess the behavior of state-of-the-art reputation-based ranking systems and, later on in this paper, to compare them with our approach. 

\subsection{Ranking context formalization}

Given a set $\mathcal U=\{u_1,\ldots,u_n\}$ of $n\in\mathbb N$ users and a set $\mathcal I=\{i_1,\ldots,i_m\}$ of $m\in\mathbb N$ items, we assume that a user $u \in \mathcal{U}$ can assign a discrete rating to an item $i \in \mathcal{I}$. 
We assume that the collected feedback is abstracted as a possibly sparse matrix of ratings denoted by $\mathcal R\in\mathbb R^{n\times m}$. 
This matrix's ratings are normalized to be in the range $]0,1]$, dividing by the maximum allowed rating. 
The difference between the maximum and the minimum normalized ratings is denoted by $\Delta_R$. When we consider a user $u\in\mathcal{U}$ and an item $i\in\mathcal{I}$, $R_{ui}=0$ if user $u$ did not rate item $i$; otherwise, $R_{ui}$ is positive.

We consider $\mathcal A =\{A_1,\ldots,A_k\}$ as a set of $k>0$ user attributes (e.g., gender, age) and let each attribute $A_j=\{a_{j_1},\ldots,a_{j_{s_j}}\}$, with $1 \leq j \leq k$, have $s_j$ classes. For instance, an attribute $A_j$ abstracting user's genders can include two or more classes, i.e., $A_j=\{male,female, ...\}$. More precisely, we denote classes of an attribute $A_j \in \mathcal A$ by $a_j, a_j', a_{j,1},\ldots,a_{j,s_j}$ and we assume that $A_j(u)=a_j$ is the class $a_j \in A_j$ for attribute $A_j \in \mathcal A$ a user $u\in\mathcal U$ belongs to. 
Finally, we identify the set of users who rated item $i\in\mathcal I$ by $\mathcal U_i=\{u\in\mathcal U\,:\,R_{ui}>0\}$, the set of items that user $u\in\mathcal U$ rated by $\mathcal I_u=\{i\in\mathcal I\,:\,R_{ui}>0\}$, and the set of users belonging to the class $a_j \in A_j$ of attribute $A_j \in \mathcal A$ by $\mathcal U(a_j)=\{u\in\mathcal U\,:\,A_j(u)=a_j\}$. 
In the context of our work, if an attribute $A_j \in \mathcal A$ has classes $A_j=\{a_{j_1},\ldots,a_{j_{s_j}}\}$, we assume that $\mathcal U(a_j)\cap\mathcal U(a_j')=\emptyset$ for all $a_j, a_j' \in A_j$, with $a_j \neq a_j'$. Furthermore, throughout this paper, given a vector $v\in\mathbb R^n$, we denote its \emph{average} by $avg(v)=\frac{1}{n}\sum_{i=1}^n v_i$ and its \emph{standard deviation} by $std(v)=\sqrt{\frac{1}{n}\sum_{i=1}^n \left(v_i-avg(v)\right)^2}$. 
Lastly, given two vectors $u,v\in\mathbb R^n$, we use the \textit{root mean squared error} (RMSE) function to evaluate how different the two vectors are: $\text{RMSE}(u,v)=\sqrt{\frac{1}{n}\sum_{i=1}^n \left(u_i-v_i\right)^2}.$ 

\subsection{Reputation-based ranking formalization}\label{sec:rep}
Our study focuses on a specific class of ranking systems whose underlying algorithm assigns a relevance score to a user, based on a notion of \emph{reputation}. Specifically, these systems aim to rank items by weighing user preferences with each user's reputation. The resulting non-personalized rankings are essential for users not logged in (e.g., course rankings in e-learning platforms, such as Udemy) or to defend the system against attacks. 
In this context, a ranking of an item $r_i$ denotes a relevance score of item $i$, based on the ratings that users assigned to the item. This is a non-personalized ranking, although it induces an order relation between items. In the scenario of recommender systems, the ranking corresponds to a set of ordered items for each user, which is then used to present a personalized recommendation list of items to that user.  

In this area, \cite{LiYHC12} proposed a reputation-based system implementing an iterative method with exponential rate convergence. The authors showed that their method is more robust to attacks than a simple arithmetic average (AA). Subsequently, \cite{saude2017robust} extended the original scheme to adjust some of its unintuitive properties: 
    \begin{itemize}
    	\item if all that rated an item $i\in I$ gave the same rating, $R_{ui}=R$, then the ranking of $i$ is almost never $R$, $r_i\neq R$, unless all those users have the same reputation;
    	\item if all that rated an item $i\in I$ gave the minimum allowed rating, $R_\bot$, then the ranking of $i$ is almost always smaller than $R_\bot$, i.e., $r_i< R_\bot$ unless all those users have the same reputation.     
    \end{itemize}
To overcome those properties, in~\cite{saude2017robust}, at each iteration, their scheme updates the ranking of each item $i$, $r_i^{k+1}$, as a weighted average of given ratings to $i$ with the reputations, $c_u^{k}$, of the users that rated the item; then, the system updates the users' reputation by computing how much the user's ratings disagree to the updated items' ranking. More precisely, their strategy can be formalized as follows\footnote{In this work, we do not consider the diversity of decay functions presented in~\cite{saude2017robust} because, in Section~8 of~\cite{saude2017robust}, the authors state that the effectiveness and robustness gains of using the different decay functions are statistically zero.}.
\begin{equation}\label{eq:gramos}
    \begin{cases}
    r_i^{k+1} = \displaystyle\sum_{u\in \mathcal U}R_{ui} c_u^k\bigg/\displaystyle\sum_{u\in \mathcal U}c_u^k\\
    c_u^{k+1} = 1 - \displaystyle\frac{\lambda}{|\mathcal I_u|}\sum_{i\in \mathcal I_u}|R_{ui}-r_i^{k+1}| 
    \end{cases}
\end{equation}
for any initial $c_u^0\in]0,1]$ (we select $c_u^0=1$) and for $\lambda\in]0,1[$, a hyper-parameter that penalizes the discordance of a user given ratings with the items' rankings. The system implementing the strategy in Eq.~\eqref{eq:gramos} not only converges with exponential rate but also is more robust to attacks than the one in~\cite{LiYHC12}. To support the reader in grasping this core concept, we present an illustrative example of the previous definitions. This example will be used to explain step-by-step the problem we address throughout this paper. 

\hlbox{Example 1 (Part 1/4)}{We consider a toy synthetic dataset with $|\mathcal U|=6$,  $|\mathcal I|=5$, $\mathcal A=\{Gender,Age\}$, with $Gender=\{A,B\}$ and $Age=\left\{]0,40],]40,\infty\right[\}$. Table~\ref{tab:Example} details the dataset, and the users' reputations and items' ranking computed with~\eqref{eq:gramos}. Users with preferences that are more different from that of the community (e.g., $u_5$) receive a lower reputation, which also affects the ranking of the items they like (e.g., $i_4$). Instead, users with the same absolute difference average to the estimated rankings (e.g., $u_2, u_3$, and $u_4$) have the same reputation, and their preferences are weighted equally in the rankings.~\hfill $\diamond$}

\begin{table}
\centering
\begin{tabular}{|c|ccccc|c|l|c|} \hline
& $i_1$ &  $i_2$ &  $i_3$ &  $i_4$ & $i_5$ & Gender & Age & $c_{u_j}$\\ \hline
$u_1$ & 1 & 0.8 & 1 & 0.4 & 0.6 & A & $]0,40]$ & 0.9255\\
$u_2$ & 1 & 1 & 1 & 0.6 & 0.6 & A & $]0,40]$ & 0.9460\\
$u_3$ & 1 & 1 & 1 & 0.6 & 0.6 & A & $]40,\infty[$ & 0.9460\\
$u_4$ & 0.8 & 1 & 1 & 0.4 & 0.6& A & $]40,\infty[$ & 0.9446\\
$u_5$ & 0.4 & 0.8 & 0.6 & 1 & 0.2 & B & $]0,40]$ & 0.8540\\
$u_6$ & 0.6 & 0.8 & 0.6 & 0.8 & 0.4 & B & $]0,40]$ & 0.9140\\
\hline
$r_{i_j}$ & 0.8071 & 0.9026 & 0.8721 & 0.6272 & 0.5052 & & & \\ \hline
\end{tabular}

\caption{Example of a synthetic dataset with users' reputations and items' rankings computed using Eq.~\eqref{eq:gramos}, with $\lambda =0.5$ ($\lambda$ should verify $\lambda\in]0,1[$) to provide a medium penalization to the users discordant with the rest of the community, and repeating the process for 8 iterations to ensure convergence.}
\label{tab:Example}
\end{table}

\paragraph{Problem formalization}
Given a set of users $\mathcal U$, a set of items $\mathcal I$, a set of ratings given by users to items $\mathcal R$, and a set of user's attributes $\mathcal A$ such that $A_j=\{a_{j_1},\ldots,a_{j_{s_j}}\} \in \mathcal A$, our goal is to: 
\begin{enumerate}
\item compute users’ reputation $\{c_u\}_{u \in \mathcal U}$ on user preferences, capturing how relevant are the preferences of an individual user for the community as a whole, in a ranking system;
\item compute rankings of items $\{r_i\}_{i \in \mathcal I}$ as a weighted average of the users’ reputations and the items’ ratings;
\item obtain reputations’ distributions that, considering every pair of $k$-tuple of classes $(a_1,\ldots,a_j)\in A_1\times\ldots\times A_k$, each associated to the set of users $\mathcal U(l=(a_1,\ldots,a_k))$, are statistically indistinguishable (\emph{multi-attribute reputation independence}).
\end{enumerate}

\subsection{Datasets}\label{sec:dataset}
Counteracting a disparate reputation on multiple attributes is not a trivial task due to the lack of public datasets with ratings and multiple users' sensitive attributes. This restriction led us to investigate this phenomenon in two real-world datasets containing both users' ratings and sensitive attributes. 

The first dataset,  \underline{Movielens-1M} (ML-1M) \citep{harper2015movielens}, contains 1,000,209 anonymous ratings of $|\mathcal I|=3,952$ movies made by $|\mathcal U| = 6,040$ users who joined MovieLens. All ratings are provided on a 5-star scale (whole-star ratings only), and each user has at least 20 ratings. User information is provided voluntarily by users, but only gender, age and job information is included in this dataset, i.e., $\mathcal A = \{gender, age,job\}$\footnote{While the dataset also offers the zip code of the users, they all are from the USA. The zip-code represents a too fine granularity, leading to too smaller groups which would not allow us to obtain statistically valid results. Hence, given that it was not possible to characterize disparate reputation based on the geographic provenience of the users and due to space constraints, the geographic perspective in Movielens is not presented.}. Specifically, the gender is denoted by a binary attribute\footnote{While gender is by no means a binary construct, to the best of our knowledge, no dataset with non-binary genders exists. What we are considering is a binary feature, as the current publicly available datasets offer.}, yielding to $\{m, f\}$. The age is specified among a set of seven ranges, originally provided together with the dataset, leading to $\{< 18, 18 - 24, 25 - 34, 35 - 44, 45 - 49, 50 - 55, > 55\}$. 
The job is specified among the following: ``other''; ``academic/educator''; ``artist''; ``clerical/admin''; ``college/grad student''; ``customer service''; ``doctor/health care''; ``executive/managerial''; ``farmer''; ``homemaker''; ``K-12 student''; ``lawyer''; ``programmer''; ``retired''; ``sales/marketing''; ``scientist''; ``self-employed''; ``technician/engineer''; ``tradesman/craftsman''; ``unemployed''; and ``writer''. 
We depict the users' distributions by attributes gender and age in Figure~\ref{fig:pie1}. From the top chart, it can be observed that male users represent the majority group, covering over 70\% of the user base. When considering the age perspective (bottom chart), the majority group is represented by the range 25-34, covering almost 35\% of the user base. The age attribute shows a strong imbalance regarding users' representation, with the age ranges covering the extremes (representing teenagers and elder people) being the less represented. 
Regarding the job attribute, the mean number of users per class is $287.619$, and the standard deviation is $223.48$. The most represented class is ``college/grad student'' ($\approx 12.57\%$) and the least represented one is ``farmer'' ($\approx 0.28\%$).

The second dataset that we used is the \underline{BookCrossing} dataset~\citep{ziegler2005improving}. It consists of 
53,408 users, 263,956 items, and  745,161 ratings. This dataset has the attributes age and location, which are provided for each user, $\mathcal A=\{age, location\}$. 
Given the $location$ of each user, originally represented as a tuple containing $(city, region, country)$ in the dataset, we created the demographic groups based on their continent of provenience. This assumption would allow us to obtain groups large enough to assess statistically valid results. To do so, we filtered the original dataset by mapping countries and their respective continent through a country-continent table\footnote{\url{https://pkgstore.datahub.io/JohnSnowLabs/country-and-continent-codes-list/country-and-continent-codes-list-csv_csv/data/b7876b7f496677669644f3d1069d3121/country-and-continent-codes-list-csv_csv.csv}}. However, this mapping is not always possible because the location data provided by users is incomplete or has spelling errors not easy to be fixed even by human curators. This process led to 22,625 users with valid continent locations. Specifically, the following continent locations were identified: 
\textbf{AF} - Africa, 
\textbf{AS} - Asia, 
\textbf{NA} - North America, 
\textbf{SA} - South America, 
\textbf{OC} - Oceania, and
\textbf{EU} - Europe. 
Again, to ensure group representations that lead to statistically valid results, we merge these locations in $\{\text{EU}, \text{AS+OC}, \text{NA+SA}, \text{AF}\}$. Finally, we decide to group the age values based on a set of four ranges, leading to $\{< 20, 20 - 40, 40 - 60, > 60\}$. These ranges ensure that demographics groups obtained by intersecting each age range with each continent location have a large enough amount of users. 

We portray the distribution of users for each of the attributes in Figure~\ref{fig:pie2}. The top chart shows that the distribution of the age groups reflects that of the MovieLens dataset; since we have fewer classes, here the majority group ($]20,40]$) covers more than half of the user base, and the young and elder users describe the tail of the distribution, being the less represented. The representation of the group, based on the geographic provenience (bottom chart), shows that the majority of users is from Europe, representing 57.61\% of the user base. The rest of the representation is mainly split between North and South America (23.13\%) and Asia and Oceania (18.58\%). As the dataset descriptions highlight, we have different sensitive attributes of the users in the two datasets, each working at a different granularity (both in the way groups are split and in the representation of the different groups of the dataset). Besides being driven by the need to guarantee statistically valid results, this situation will also pose us in a better position to assess our problem and the effectiveness of our approach in various real-world settings.
\begin{figure}[!t] 
\centering      
\subfigure[]{\includegraphics[width=0.4\textwidth]{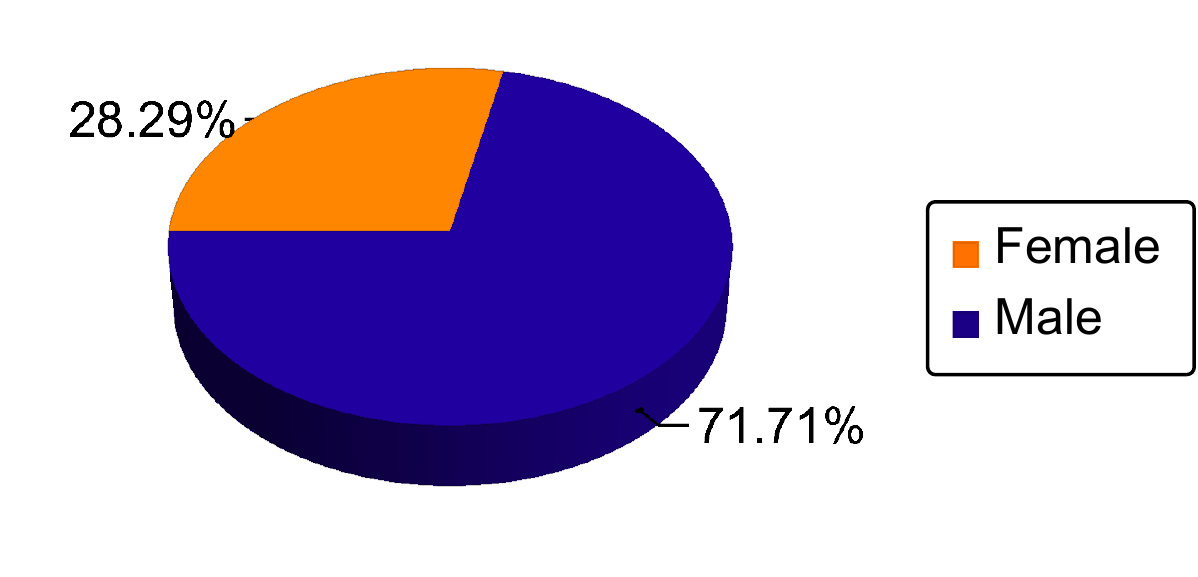}}
\qquad
\subfigure[]{\includegraphics[width=0.4\textwidth]{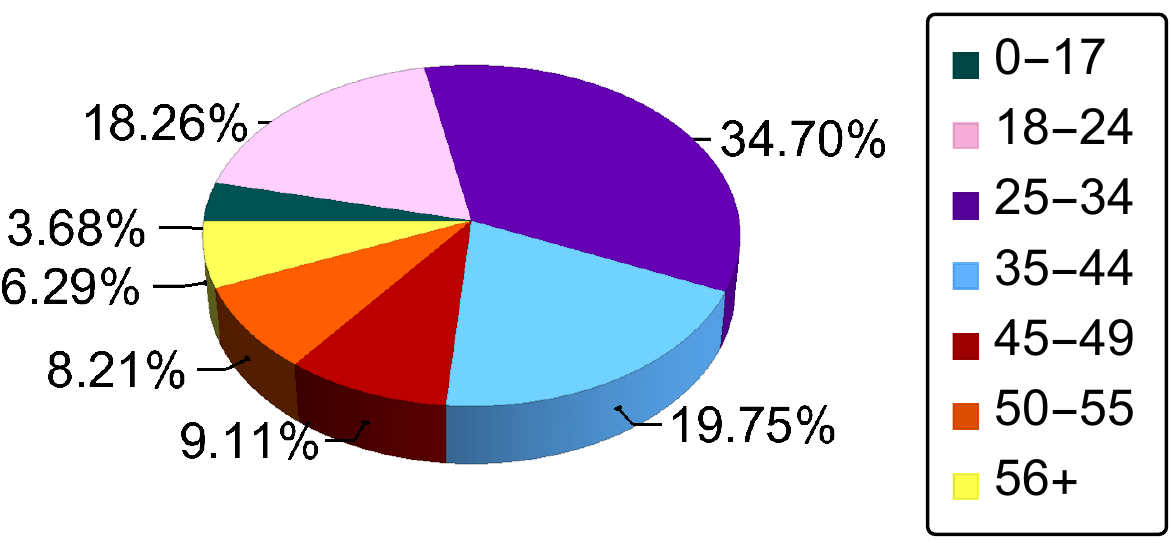}}
\caption{The distribution of the ML-1M dataset users based on gender~(a) and age~(b).}
\label{fig:pie1} 
\end{figure}

\begin{figure}[!t] 
\centering  
\subfigure[]{\includegraphics[width=0.4\textwidth]{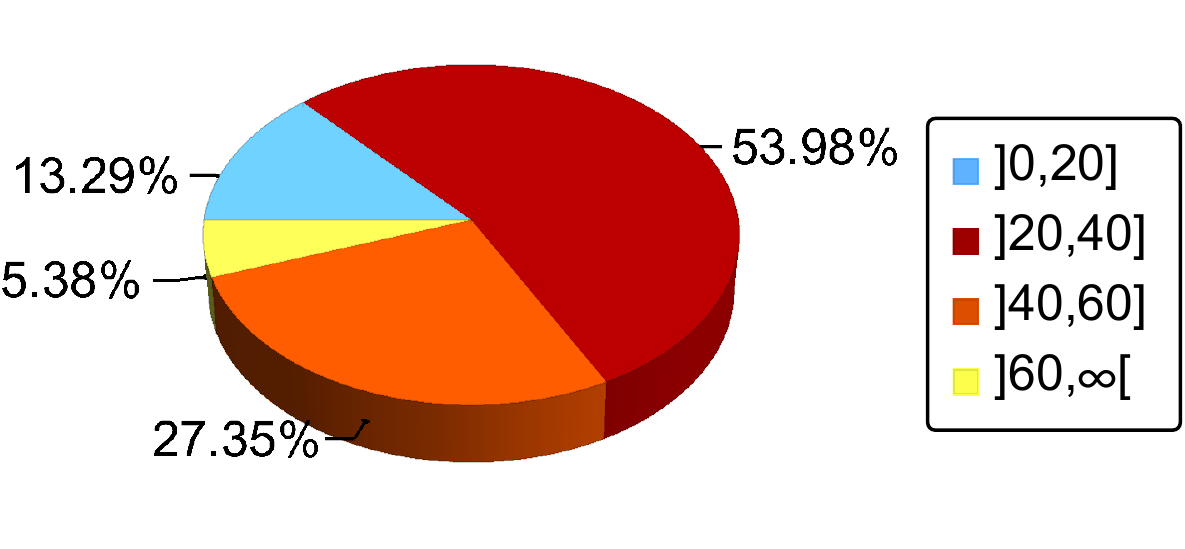}}
\qquad
\subfigure[]{\includegraphics[width=0.4\textwidth]{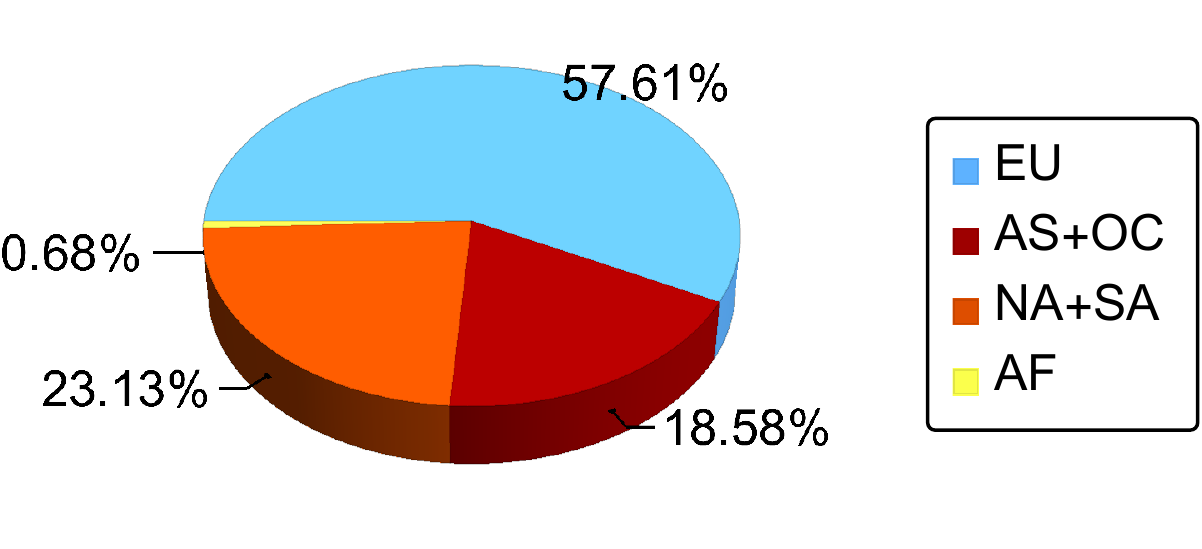}}

\caption{The distribution of the BookCrossing dataset users based on age~(a) and location~(b).}
\label{fig:pie2} 
\end{figure}

It is worth noticing that the users' attributes that we are considering are limited to what is available in reported benchmark datasets. As future work, we would like to evaluate our results in novel datasets containing more users' attributes, or even collect datasets that allow to overcome this limitation.

\section{Impact of Single-Attribute Reputation Independence on Different Protected Groups}\label{sec:single}
This section presents recent advances in mitigating reputation disparities among demographic groups, aiming to guarantee the independence of reputation scores from users' sensitive attributes. 

\subsection{Single-Attribute Mitigation Methodology}
\cite{RamosB20} introduced the \emph{Disparate Reputation} (DR) concept in reputation-based ranking systems, as the difference between the average reputation of the users belonging to two distinct classes. 
Given different classes $a$ and $b$ for the same attribute, the disparate reputation $\Delta(a,b)$ is:
 $$ \Delta(a,b)=\mathcal \mu_{a}-\mu_{b},$$
 where $\mu_a=avg(\{c_u\,:\,u\in\mathcal U(a)\})$ and $\mu_b=avg(\{c_u\,:\,u\in\mathcal U(b)\})$.  
The metric ranges in $[-1+\Delta_R\lambda,1-\Delta_R\lambda]$,  recalling that $\Delta_R$ is the difference between the maximum and the minimum normalized ratings. Its value is $0$ when both averages of the reputations are the same ($\mu_a=\mu_b$). Negative values point that class $b$ has users with higher reputation values and, vice-versa, for the class $a$ and positive values. 

To characterize if disparate reputation systematically affects the users belonging to a class, the authors proposed to perform a Mann-Whitney (MW) statistical test. The MW statistical test was performed between each pair of user groups' reputation distributions relative to an attribute. We may use this strategy to check whether two independent samples come from populations with the same distribution. 
This test is performed on each pair of groups and compares the median of the two samples. In this work, we use instead the \textit{two-sample location test (LT)} to compare the means of two samples, to enable a more direct comparison between DR and LT because both use the means. 

To mitigate disparate reputation when considering a single users' sensitive attribute (that does not need to be binary), the authors performed a final post-processing step. Specifically, after the iterative method in Eq.~\eqref{eq:gramos} (i.e., after $N$ iterations), this additional step is defined as follows.

\begin{equation}\label{eq:fair}
    \begin{cases}
    c_u = \mu +\displaystyle \left(c_u^N-\mu_l\right)\frac{\sigma}{\sigma_l}, & \text{for }l=1,\ldots,k\textit{ and }u\in \mathcal U(a_l)\\
    r_i = \displaystyle\sum_{u\in \mathcal U}R_{ui} c_u\bigg/\displaystyle\sum_{u\in \mathcal U}c_u &
    \end{cases},
\end{equation}
where $\mu=avg(\{c_u\,:\,u\in\mathcal U\})$, $\sigma=std(\{c_u\,:\,u\in\mathcal U\})$, $\mu_l=avg(\{c_u\,:\,u\in\mathcal U(a_l)\})$, and $\sigma_l=std(\{c_u\,:\,u\in\mathcal U(a_l)\})$. 
Furthermore, $c_u$ denotes the final reputation of user $u$ and $r_i$ the final ranking of item $i$. Concretely, our adjusted reputation scores, $c_u$, enforces the reputations obtained from Eq.~\ref{eq:gramos} with the same distribution for each group of users under the sensitive attributes.

The authors showed that the reputations' distributions for each class of a sensitive attribute become statistically indistinguishable, leading to single-attribute reputation independence after this additional step. 
Though this solution can mitigate a reputation bias for user groups based on a given sensitive attribute, it remains unclear whether an analogous bias on other attributes ends up being mitigated by considering a unique attribute. This problem motivated us to perform a more extensive evaluation of the described methodology in this paper. 
Hence, we explore the methodology detailed above in Example~1.  

\hlbox{Example 1 (Part 2/4)}{Recalling Part 1, if we consider the attribute $Gender$ for instance, we observe that the average reputations for users in each class is $\mu_{A}\approx 0.9405$ and $\mu_B \approx 0.8840$. Hence, the DR is $\Delta(A,B)=\mu_{A}-\mu_{B}=0.0565$. 
Therefore, on average, the opinion of users with $Gender$ $ A $ is more influential for the system while computing the items' rankings. 
Given that this is a toy example, the LT statistical test may not be used. Using the additional step of Eq.~\eqref{eq:fair}, we obtain the users' reputations and items' rankings in Table~\ref{tab:Example_2}. Specifically, we obtain $\tilde{\mu}_A=\tilde{\mu}_B\approx 0.8840$, yielding a DR of $\widetilde{\Delta(A,B)}\approx 0$. 
Now, on average, the opinion of users with $Gender$ $A$ is equally reflected as that of the users with $Gender$ $B$.~\hfill $\diamond$}

\begin{table}
\centering
\begin{tabular}{cccccc}
  $c_{u_1}$ & $c_{u_2}$ & $c_{u_3}$ & $c_{u_4}$ & $c_{u_5}$ & $c_{u_6}$
\\ \hline
0.8690 & 0.8895 & 0.8895 & 0.8881 & 0.8769 & 0.8911 \\
\end{tabular}
\begin{tabular}{ccccc}
  $r_{i_1}$ & $r_{i_2}$ & $r_{i_3}$ & $r_{i_4}$ & $r_{i_5}$ 
\\ \hline
0.8001& 0.9006& 0.8667& 0.6335& 0.5003 \\
\end{tabular}

\caption{Users' reputations and items' ranking of Example~1 after performing the additional step in Eq.~\eqref{eq:fair} for the gender sensitive attribute. Similar patterns show up for age groups.}
\label{tab:Example_2}
\end{table}

\subsection{Exploratory Analysis on Reputation under Multiple Sensitive Attributes}

We start by doing an exploratory analysis of the reputation-based system formalized in Section~\ref{sec:rep} of Example~1. 
In this analysis, the goal is to understand if, when mitigating bias for a sensitive attribute, there is still a bias related to another sensitive attribute. 

\hlbox{Example 1 (Part 3/4)}{After mitigating bias for attribute $Gender$ in Part 2, we now compute the DR for the attribute $Age$. 
We have that $\Delta\left(]0,40],]40,\infty[\right)=\mu_{]0,40]}-\mu_{]40,\infty[}\approx 0.8816-0.8888=-0.0072$. 
So, on average, the opinion of users with $Age$ $]40,\infty[$ is more reflected in items' rankings, and there is still a bias related to attribute $Age$.~\hfill $\diamond$}

To emphasize the existence of this issue in a real-world context, we then conduct an exploratory analysis on the reputation-based system formalized in Section~\ref{sec:rep}, under the ML-1M dataset~\citep{harper2015movielens}. 
The latter dataset includes gender (binary for this dataset, but not binary in general) and age as user's sensitive attributes. More details can be found in Section \ref{sec:dataset}.  

First, we apply the mitigation strategy described in Eq.~\eqref{eq:fair}, grouping users based on their gender. 
Illustrated in Figure~\ref{fig:genderBWC}, the results on disparate reputation confirm that introducing the additional step formalized in Eq. \ref{eq:fair} in the reputation-based ranking methodology leads to users' reputations independence for gender-based groups. 
Notwithstanding, under the same scenario, if we group reputation scores based on another sensitive attribute -- the age -- and measure the disparate impact on the resulting reputation distributions, then there is a disparate reputation for the attribute age, as revealed in Figure~\ref{fig:gender_ageBWC}. 
More precisely, by applying Eq.~\eqref{eq:gramos} and Eq.~\eqref{eq:fair} sequentially, it possible to mitigate a reputation bias on the attribute gender. However, as depicted in Figure~\ref{fig:gender_ageBWC}, the reputation bias on age groups is not mitigated in a collateral fashion. 
The sub-figures of Figure~\ref{fig:gender_ageBWC} indicate the reputations on age-based groups before and after mitigating for attribute gender, showing identical values. 
This example confirms that mitigation on genders does not mitigate on age ranges collaterally, since the reputation scores are still biased for age-based groups.

Similar experiments were conducted with Eq.~\eqref{eq:fair} to mitigate reputation bias on age and test gender reputation bias. None of the alternatives showed that mitigating a reputation bias on one attribute also works on other attributes simultaneously. Hence, we can draw the following proposition:
 
\begin{figure}[!t] 
\centering      
\subfigure[reputation scores]{\includegraphics[width=0.40\textwidth]{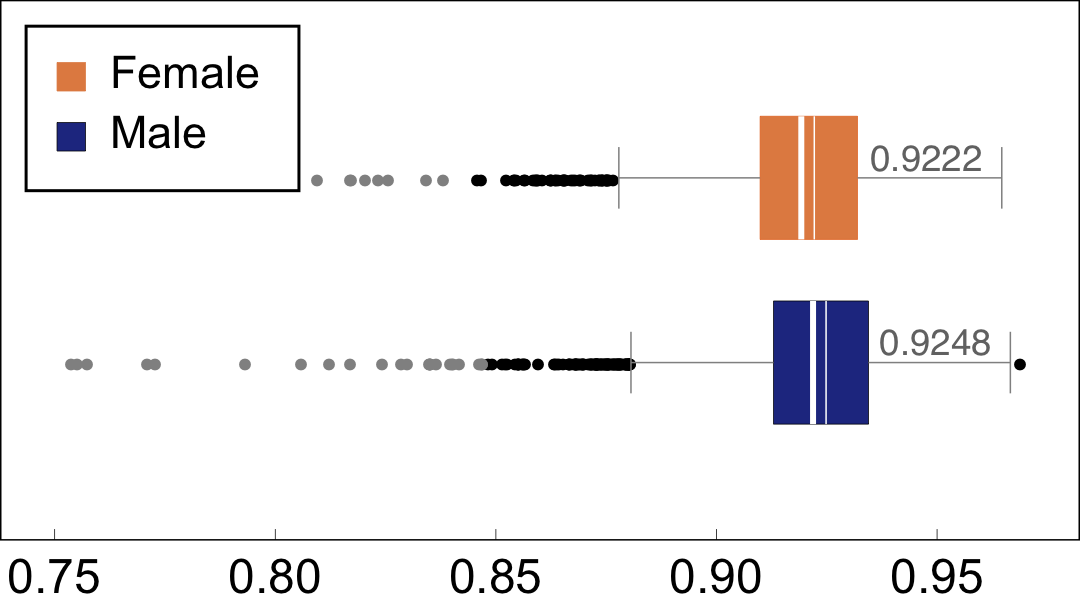}}
\qquad
\subfigure[reputation scores]{\includegraphics[width=0.40\textwidth]{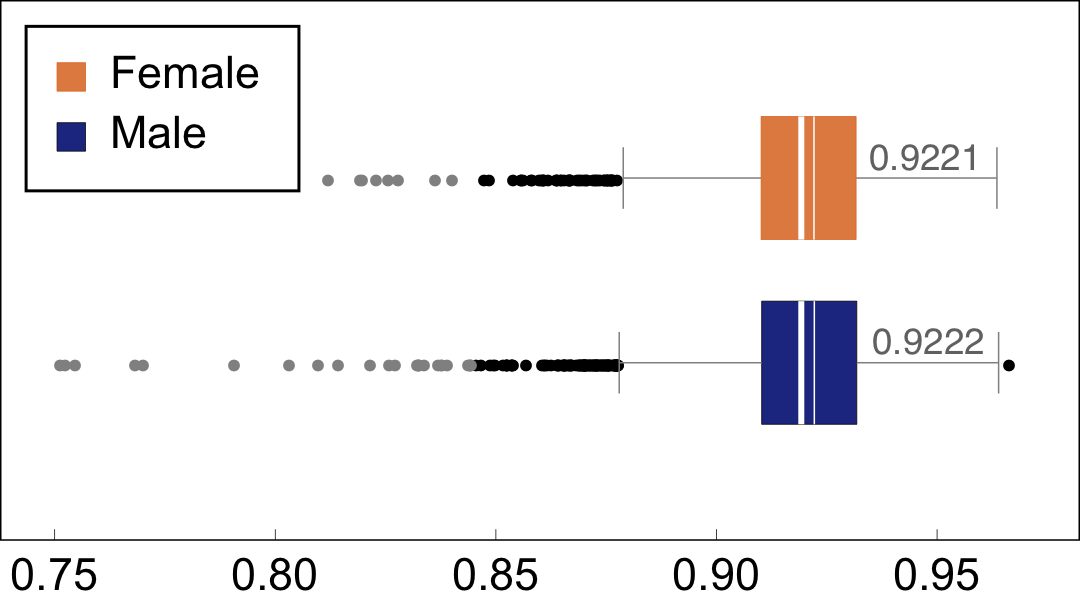}}

\caption{Box-whisker-chart for users' reputations resulting from Eq.~\eqref{eq:gramos} (a), and from Eq.~\eqref{eq:gramos} and Eq.~\eqref{eq:fair} (b), with $\lambda = 0.5$, for user groups based on gender (ML-1M).}
\label{fig:genderBWC} 
\end{figure}

\begin{figure}[!t] 
\centering      
\subfigure[reputation scores]{\includegraphics[width=0.40\textwidth]{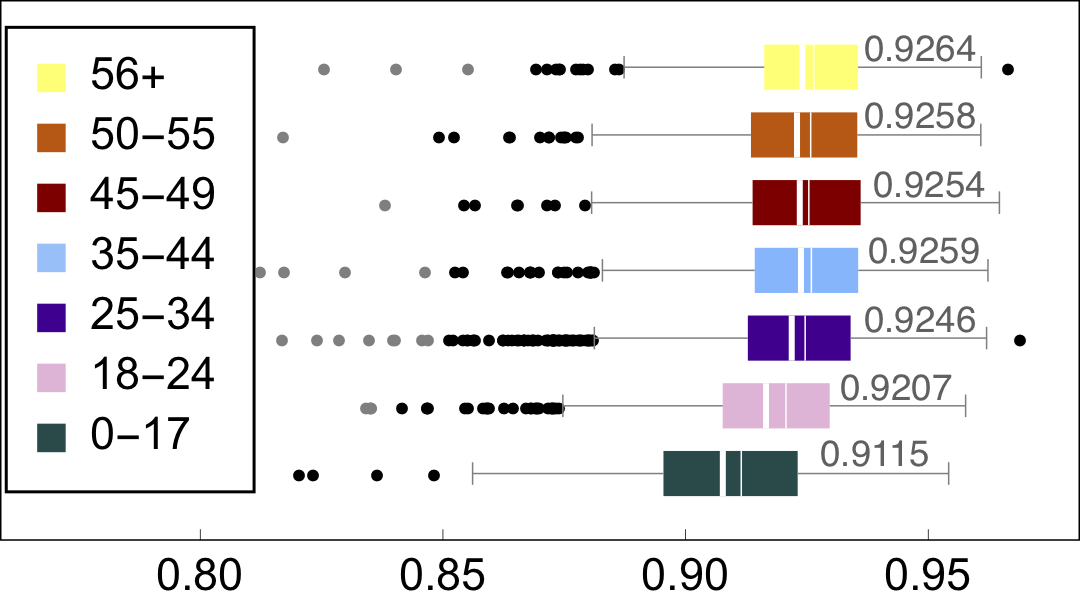}}
\qquad
\subfigure[reputation scores]{\includegraphics[width=0.40\textwidth]{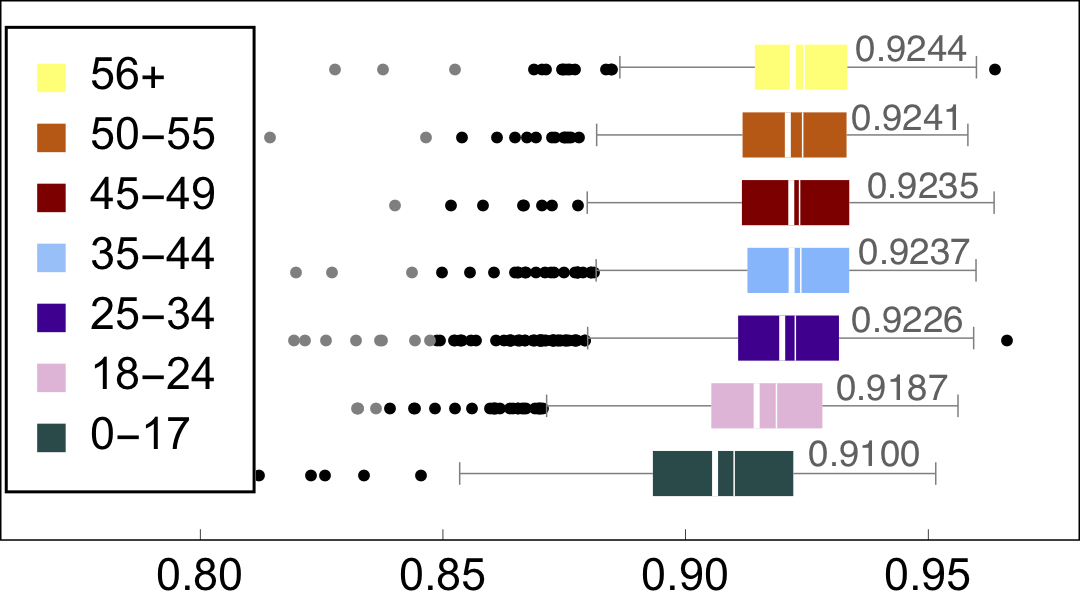}} 
\caption{Box-whisker-chart for reputations of users resulting from Eq.~\eqref{eq:gramos} (a), and from Eq.~\eqref{eq:gramos} and Eq.~\eqref{eq:fair} \underline{applied to  attribute gender and evaluated on attribute age} (b), with $\lambda = 0.5$, for user groups based on age (ML-1M).} 
\label{fig:gender_ageBWC} 
\end{figure}

\begin{prop}\label{prop:1}
Given a set of users $\mathcal U$, a set of items $\mathcal I$, a set of ratings that users gave to items $\mathcal R$ and a set of user attributes $\mathcal A = \{A_1,\ldots,A_k\}$, mitigating a reputation bias with Eq.~\eqref{eq:fair} for each attribute individually (for any order) does not necessarily yield reputations without bias for both attributes.~\hfill$\circ$ 
\end{prop}
\begin{proof}
The proof is based on providing a counter-example, considering the MovieLens-1M dataset. 
Consider $k=2$, i.e., only two sensitive attributes, i.e., gender and age. Consider applying Eq.~\eqref{eq:fair} to mitigate a reputation bias for the attribute gender first, and after for the attribute age. The final reputations' averages for the gender are statistically different, $\mu_{female}=0.906088$ and $\mu_{male}=0.906067$, hence $\mu_{female}\not\approx \mu_{male}$.
\end{proof}

Given such a finding uncovered by this paper, it becomes of utmost importance to investigate whether it is possible to devise a method that reduces a reputation bias for multiple attributes jointly. 

\section{Multi-attribute Reputation Independence}\label{sec:mit_mul}
In this section, we aim to avoid the ranking system being systematically impacted by a bias against groups under a sensitive attribute while mitigating a bias against groups for another attribute. To this end, we design a strategy that, given a set of users' sensitive attributes mitigates the user reputations' bias against user groups characterized by different combinations of those attributes. This strategy allows us to make the reputation computation independent from all the considered sensitive attributes, at the same time. We properly designed each component of our method to ensure both the feasibility and efficiency of reputation-based ranking systems. We can easily extend our approach to embrace more than two sensitive attributes, at the same time.  

The method proposed in this paper is based on partitioning users according to more than one attribute, jointly. Specifically, let $\mathcal A=\{A_1,\ldots,A_k\}$ be a set of $k>0$ attributes and let each attribute $A_j=\{a_{j,1},\ldots,a_{j,s_j}\}$ have $s_j$ classes. 
Now, we consider all the $k$-tuples of classes $(a_1,\ldots,a_k)\in A_1\times\ldots\times A_k$.  
Subsequently, to each $k$-tuple, we associate the set of users $\mathcal U(l=(a_1,\ldots,a_k))$, which is the set of users such that $a_j\in A_j$ for $j=1,\ldots,k$. 
For instance, in Example~1 (Table~\ref{tab:Example}), we have that $\mathcal U(l=(A,]40,\infty[))=\{u_3,u_4\}$. 
It should be noted that the sets of users for all the possible $k$-tuples form a partition of $\mathcal U$, as desired. 
Denoting by $c_u^N$ the outcome reputation of user $u$ after running Eq.~\eqref{eq:gramos} for $N$ iterations, we arrange Eq.~\eqref{eq:fair} as follows.
\begin{equation}\label{eq:multi_fair}
    \begin{cases}
    c_u = \mu +\displaystyle \left(c_u^N-\mu_l\right)\frac{\sigma}{\sigma_l}, & \text{for }u\in \mathcal U(l=(a_1,\ldots,a_k))\text{ and}\\
    & (a_1,\ldots,a_k)\in A_1\times\ldots\times A_k\\
    r_i = \displaystyle\sum_{u\in \mathcal U}R_{ui} c_u\bigg/\displaystyle\sum_{u\in \mathcal U}c_u &
    \end{cases}
\end{equation}

\noindent where, for $l\in A_1\times\ldots\times A_k$, $\mu = \displaystyle\mathop{\min}_{l}\mu_l$ and $\sigma = \displaystyle\mathop{\min}_{l}\sigma_l$, with $\mu=avg(\{c_u\,:\,u\in\mathcal U\})$, $\sigma=std(\{c_u\,:\,u\in\mathcal U\})$, $\mu_l=avg(\{c_u\,:\,u\in\mathcal U(l)\})$, and $\sigma_l=std(\{c_u\,:\,u\in\mathcal U(l)\})$. Furthermore, $c_u$ denotes the final reputation of user $u$ and $r_i$ the final ranking of item $i$. 
It should be also noted that, in Eq.~\eqref{eq:multi_fair}, the option of choosing the minimum of the averages and the minimum of the standard deviations ensures that the reputations' re-scaling lies in the interval $]0,1]$. Similarly to Eq.~\eqref{eq:fair}, even Eq.~\eqref{eq:multi_fair} reconciles reputation's distributions for each $k$-tuple of attributes classes so that the reputations of each $k$-tuple of classes are ``statistically indistinguishable''. This finally leads to the targeted \emph{multi-attribute reputation independence}.

\hlbox{Observation 2}{The post-processing step proposed in this paper in Eq.~\eqref{eq:multi_fair} can be included in any ranking system that calculates rankings by means of a weighted average of the ratings, assuming that the weights represent users' reputation scores.~\hfill $\diamond$}

Notice that the proposed method ensures that each demographic  group of users sees its average opinion reflected proportionally to the size of the group. This would not be the case if there exists disparate reputation between the demographic groups.
In order to help the reader in grasping the idea conveyed by means of Eq.~\eqref{eq:multi_fair}, we apply it to the setting described in Example~1. 

\hlbox{Example 1 (Part 4/4)}{The proposed Eq.~\eqref{eq:multi_fair} is used to mitigate a bias affecting groups characterized by multiple sensitive attributes. First, we build the sensitive meta-attributes of Example~1: 
$\mathcal A=\left\{\left(A,]0,40]\right),\left(A,]40,\infty[\right),\left(B,]0,40]\right)\right\}$. 
Then, by applying Eq.~\eqref{eq:multi_fair} to these meta-attribute groups, we obtain the following average reputations: $\mu_A=\mu_B=\mu_{]0,40]}=\mu_{]40,\infty[}=0.8840$. Therefore the DR reaches zero for every pair of sensitive attributes.~\hfill $\diamond$}

Following the intuition apprehended from Example~1, we then demonstrate the soundness of the proposed method. 

\begin{theorem}\label{th:main}
Consider a matrix of ratings $\mathcal R$, with set of items $\mathcal I$, set of users $\mathcal U$ and set of $k$ users' sensitive attributes $\mathcal A=\{A_1,\ldots,A_k\}$. 
Let the users' reputations and items' rankings be computed with Eq.~\eqref{eq:gramos}. If we apply Eq.~\eqref{eq:multi_fair} to recompute users' reputations and items' rankings, using the attributes $A_1,\ldots,A_k$, then the following property holds: for any two classes of any two sensitive attributes, $a\in A_i$ and $a'\in A_j$ ($A_i,A_j\in\mathcal A$), the set of users $\mathcal U(a)$ reputations and the set of users $\mathcal U(a')$ reputations have zero disparate reputation ($\mu_a=\mu_a'$).\hfill$\circ$
\end{theorem}
\begin{proof}
First, for any two classes of any two sensitive meta-attributes $\tilde{a}\in  A_1\times\ldots\times A_k$ and $\tilde{a}'\in  A_1\times\ldots\times A_k$, Eq.~\eqref{eq:multi_fair} makes the set of users $\mathcal U(\tilde{a})$ reputations and the set of users $\mathcal U(\tilde{a}')$ reputations have zero disparate reputation ($\mu_{\tilde{a}}=\mu_{\tilde{a}'}=\mu$). Next, we observe that: 
$$\mathcal U(a)=\left\{u\in\mathcal U(a_1,\ldots, a_k)\,:\, a_1\in A_1,\ldots,a_k\in A_k\text{ and }a_i=a\right\}$$ 
and
$$\mathcal U(a')=\left\{u\in\mathcal U(a_1,\ldots, a_k)\,:\, a_1\in A_1,\ldots,a_k\in A_k\text{ and }a_j=a'\right\}.$$
For each set of users $\mathcal U(a_1,\ldots, a_k)$, as noted before, the average reputation is the same, $\mu$. 
It remains to demonstrate that the average of a finite collection of finite sets with the same average $\mu$ is also $\mu$. 
To ease the notation, consider the sets of users $U_1,\ldots,U_k$ that have the same average reputation $\mu$. Subsequently, we observe that:
\begin{align*}
    \frac{1}{\sum_{i=1}^k|U_i|}\displaystyle\displaystyle\sum_{i=1}^k \displaystyle\sum_{u\in U_i} c_u  & = \frac{1}{\sum_{i=1}^k|U_i|}\displaystyle\sum_{i=1}^k |U_i|\frac{\sum_{u\in U_i} c_u}{|U_i|}\\
    & = \frac{\mu \sum_{i=1}^k |U_i|}{\sum_{i=1}^k|U_i|}= \mu
\end{align*}
This concludes the argument, and $\mu_a=\mu_{a'}=\mu$.
\end{proof}

\begin{prop}\label{prop:complex}
Given a set of users $\mathcal U$, a set of items $\mathcal I$, a set of ratings that users gave to items $\mathcal R$ and a set of user attributes $\mathcal A = \{A_1,\ldots,A_k\}$, the time-complexity of computing the iterative scheme in Eq.~\eqref{eq:gramos} for $N>0$ iterations followed by Eq.~\eqref{eq:multi_fair} is $\mathcal O\left(N |\mathcal U||\mathcal I|+k|\mathcal U|\right)$.~\hfill$\circ$
\end{prop} 
\begin{proof}
First, we perform Eq.~\eqref{eq:gramos} for $N$ iterations, with each iteration composed of two steps. The first step computes the rankings of $|\mathcal I|$ items as a weighted average of the users' ratings by the users' reputations, which takes $\mathcal O(|\mathcal U|)$ time.  
For one iteration, the rankings' update requires $\mathcal O(|\mathcal I||\mathcal U|)$. The second step of each iteration requires to compute reputations, given that there are $|\mathcal U|$ users for an update of their reputation is needed. Updating the reputation of a user $u$ implies running $\mathcal I_u$ operations. which are $\mathcal I_u=\mathcal I$ in the worse case. The second step has $\mathcal O(|\mathcal I||\mathcal U|)$ operations as well. For $N$ iterations, the time-complexity is $\mathcal O(N|\mathcal I||\mathcal U|)$. We then compute Eq.~\eqref{eq:multi_fair} in two steps. The first step of Eq.~\eqref{eq:multi_fair} requires updating $|\mathcal U|$ users' reputations. 
Computing $\mu_l$ and $\sigma_l$ can be done linearly in the number of users in $\mathcal U(l)$. Since $\bigcup_{l} \mathcal U(l)=\mathcal U$ is a users' partition, we can compute all the $\mu_l$ and $\sigma_l$ in $\mathcal O(|\mathcal U|)$. 
Additionally, we may verify if a user $u\in\mathcal U(l=(a_1,\ldots,a_k))$ in $\mathcal O(k)$, by checking the $k$ classes values of the $k$ attributes in a table indexed by the users.  
The second step of Eq.~\eqref{eq:multi_fair} updates the items' rankings, which again takes $\mathcal O(|\mathcal I||\mathcal U|)$ time. 
Hence, Eq.~\eqref{eq:multi_fair} requires $\mathcal O(k |\mathcal U|+|\mathcal I||\mathcal U|)$ steps. The total time-complexity is, thus, $\mathcal O\left(N |\mathcal U||\mathcal I|+k|\mathcal U|\right)$.
\end{proof}

In general, the number of users' attributes is smaller than the number of items. In this scenario, the time-complexity of Proposition~\ref{prop:complex} may be simplified to $\mathcal O\left(N |\mathcal U||\mathcal I|\right)$, i.e., the same order of complexity of running solely Eq.~\eqref{eq:gramos}.

\section{Experimental Evaluation}\label{sec:evaluation}
In this section, we evaluate our multi-attribute mitigation approach in order to answer three key research questions: 

\vspace{2mm} \noindent \textbf{RQ1} Does our method mitigate bias over multiple attributes jointly?

\vspace{1mm} \noindent \textbf{RQ2} Does our method preserve system robustness against attacks? 

\vspace{1mm} \noindent \textbf{RQ3} How is ranking quality affected by the disparate reputation? 

\subsection{Metrics}

\vspace{2mm} \noindent \textit{Disparate Reputation~\citep{RamosB20}}. Let $A_j \in \mathcal A$ be an attribute of the users, with more than one class $A_j = \{a_{j,1}, \ldots , a_{j,s_j}\}$. Considering a class $a_j \in A_j \in \mathcal A$, we denote as $\mu_{a_j} = avg(\{c_u\}_{u \in \mathcal U_{a_j}})$ the average reputation of the users characterized by that class; this average reputation is used a proxy of user group reputation. The corresponding disparate reputation metric is computed as the difference between two averaged user group reputations, i.e., $\Delta(a_j,a_j') = \mu_{a_j} - \mu_{a_j'}$. The disparity is 0 when reputation averages are the same ($\mu_{a_j} = \mu_{a_j'}$). Negative values point that class $\mu_{a_j} $ has users with higher reputation values and, vice-versa, for the class $\mu_{a_j'}$ and positive values. 
The central role of the proposed work is to ensure for each demographic group the reputations among users follow identical distributions (in the statistical sense). Therefore, the disparate reputation metric measures how different are the averages of reputation distributions of different demographic groups.

\vspace{2mm} \noindent \textit{Robustness~\citep{saude2017robust}}. Let $r=(r_1,\ldots,r_m)$ be the vector with item rankings in the absence of attacks and $r_{attacked}=(r_1',\ldots,r_m')$ be the vector with items' rankings in the presence of attacks. 
To evaluate the robustness, we use the Kendall Tau ($\tau$) metric, applied to the rankings obtained without attackers against the rankings obtained when considering attackers. Specifically, we define the robustness as $\tau(r,r_{attacked})$. This metric monitors the ordinal association between two quantities. Intuitively, the Kendall correlation between two variables is higher when observations are identical and lower otherwise. Scores close to 0 mean that the system is vulnerable to attacks; conversely, scores increase in tandem with robustness. 
 
\begin{figure}[!hb]  
\centering      
\subfigure[reputation scores]{\includegraphics[width=0.4\textwidth]{1attGender1_.png}}\qquad
\subfigure[reputation scores]{\includegraphics[width=0.4\textwidth]{2att3_.png}}
\\
\subfigure[reputation scores]{\includegraphics[width=0.4\textwidth]{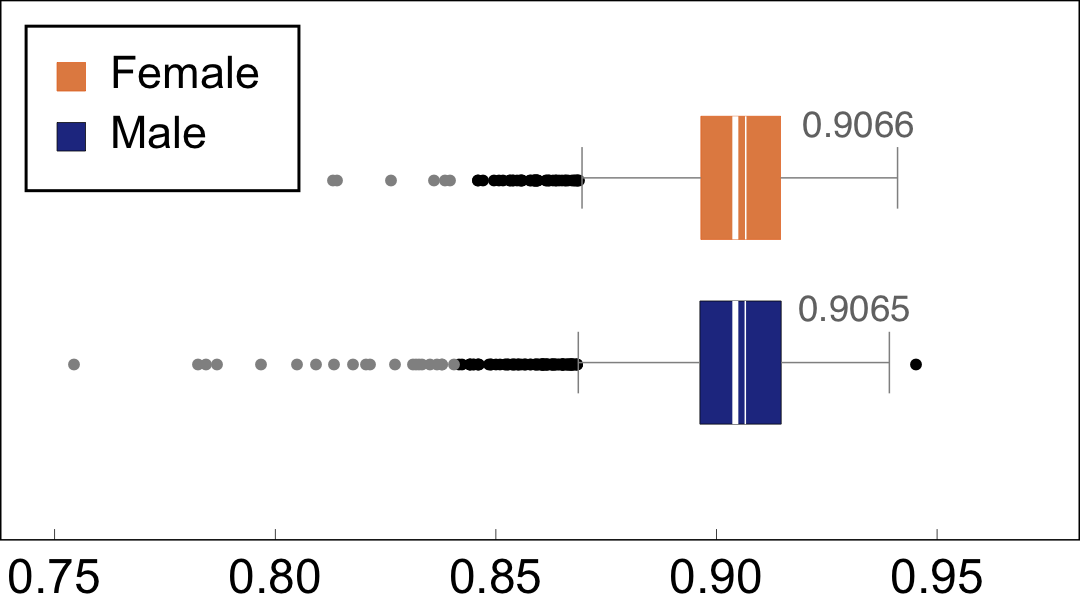}}\qquad
\subfigure[reputation scores]{\includegraphics[width=0.4\textwidth]{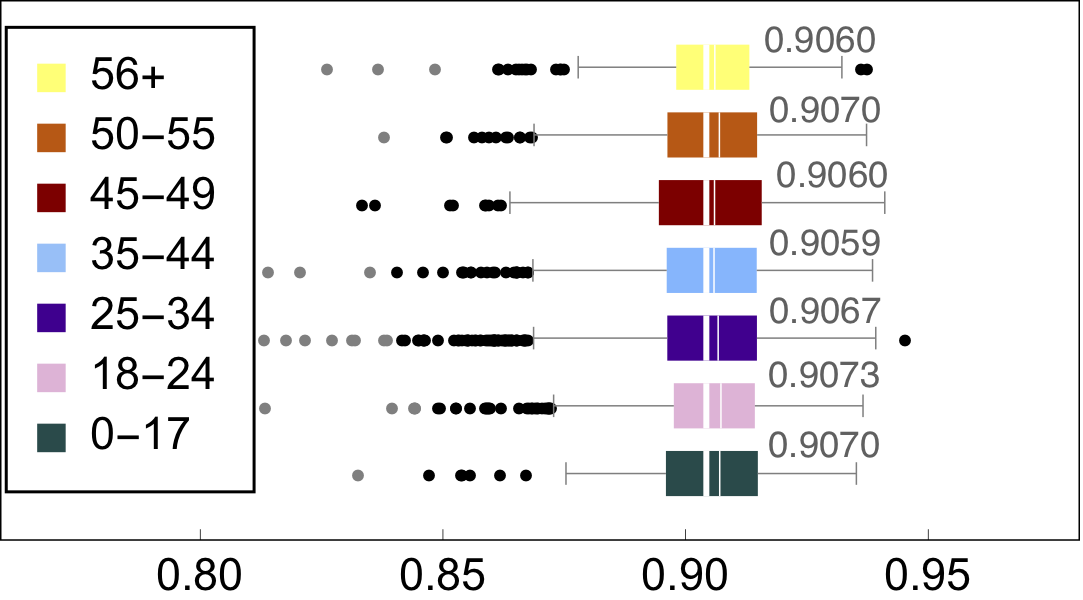}}

\caption{Box-whisker-chart for reputations of users resulting from Eq.~\eqref{eq:gramos} in (a) and (b), respectively, and from Eq.~\eqref{eq:gramos} plus Eq.~\eqref{eq:multi_fair} in (c) and (d), respectively, with $\lambda = 0.5$, for the multiple attributes \underline{gender} and \underline{age} (ML-1M).}
\label{fig:genderAgeBWC} 
\end{figure}

\subsection{Disparate Reputation Evaluation (RQ1)}\label{sec:ev_disp_rep}

To answer this question, we investigate the extent to which a bias on users’ reputations exists, when two sensitive attributes are considered at the same time.  
Nonetheless, the following property holds.

\hlbox{Observation 3}{Our approach, presented in Section~\ref{sec:mit_mul}, works with {\em any} set of sensitive attributes; hence, the cardinality of the set $\mathcal A$ can be higher than 2. Due to the limitations of the existing datasets, this paper focuses on the case of two sensitive attributes.~\hfill $\diamond$}

In what follows, we first present results for each dataset individually, then connect together all the findings in the final discussion.

\vspace{1mm} \noindent \textbf{MovieLens-1M}.
We first characterize the disparate reputation, after applying the original method in Eq.~\eqref{eq:gramos}. The results are reported in a Box-whisker-chart (BWC) representing the average reputation for each group, considering gender and age as attributes. \figurename~\ref{fig:genderAgeBWC}~(a) shows us that using solely Eq. \ref{eq:gramos} leads to a consistent reputation disparity on the gender-based groups. Specifically, on average, male users have higher reputation values than female users, yielding to a gender bias. 
Then, we test the null hypothesis for both attributes that the mean difference is 0 at the 5\% level based on the LT test. 
The hypothesis is rejected, confirming a gender bias.

\begin{table}[!t]
\centering
\begin{tabular}{|c|ccccccc|} \hline
 & $<18$ & $18-24$ & $25-34$ & $35-44$ & $45-49$ & $50-55$ & $>55$
\\ \hline
$<18$ &  -- & \textbf{-0.0090} & \textbf{-0.0142} & \textbf{-0.0161} & \textbf{-0.0159} & \textbf{-0.0153} & \textbf{-0.0164} \\
$18-24$ & -- & -- & \textbf{-0.0053} & \textbf{-0.0072} & \textbf{-0.0070} & \textbf{-0.0064} & \textbf{-0.0075}  \\
$25-34$ & -- & -- & -- & \textbf{-0.0019} & -0.0017 & -0.0011 & -0.0022 \\
$35-44$ & -- & --& -- & -- & 0.0002 & 0.0008 & -0.0003 \\
$45-49$ & -- & -- & -- & -- & -- & 0.0006 & -0.0005 \\
$50-55$ & -- & -- & -- & -- & -- & -- & -0.0011 \\ \hline
\end{tabular}

\caption{DR of reputations resulting from Eq.~\eqref{eq:gramos}, for attribute age (ML-1M). 
The LT tests for the reputation scores are denoted as normal text if $H_0$ is not rejected and bold text if $H_0$ is rejected.}
\label{tab:3}
\end{table}

Figure~\ref{fig:genderAgeBWC}~(b) uncovers a consistent reputation disparity on the attribute age, when applying only Eq. \ref{eq:gramos}. 
Users belonging to younger groups have, on average, a lower reputation than older users, leading to a bias on the attribute age. 
The disparate reputation metric, when only Eq.~\eqref{eq:gramos} is used, yields the results in Table \ref{tab:3}, which reveal a prominent bias. Table~\ref{tab:3} also reports the LT test for users' reputations, assessing if the null hypothesis that the mean difference is 0 ($H_0$) or not ($H_1$) at a 5\% confidence level. 
We filled only the up-triangular part of the table, since the DR metric anti-commutes (and the LT commutes); the low-triangular part is equal to the symmetric of the up-triangular one. (the low-triangular part is equal to the up-triangular one).

When we mitigate bias for both the gender and age attributes with Eq.~\eqref{eq:multi_fair}, we obtain the BWC for reputations under the gender attribute of Figure~\ref{fig:genderAgeBWC}~(c). Under this setting, we get a disparate reputation of $\Delta(a,a') \approx 0$, for each pair of gender-based groups, finally mitigating the bias on the attribute gender. This time, the null hypothesis that the mean difference is 0 is not rejected at the $5\%$ confidence level, based on the LT test. This result confirms that we mitigated the bias on the reputations for these two classes. At the same time, for the attribute age, we achieve the results in 
Figure~\ref{fig:genderAgeBWC}~(d). 
Now, the null hypothesis that the reputations’ mean difference is 0 ($H_0$) is not rejected at the 5\% confidence level, using the LT test, for any pair of age classes.

We then tested the proposed approach on three attributes: gender, age and job. We filtered the obtained groups of users to select only those with more than two users. By doing so, we obtained 195 groups of users. The Box-whisker-charts in  \figurename~\ref{fig:genderAgeJobBWC} (a) and (b) respectively report the results of the reputation-based ranking system in Eq.~\eqref{eq:gramos} (i.e., without accounting for disparate reputation) and those of the system deriving from Eq.~\eqref{eq:multi_fair}, introducing reputation independence. Given the large number of groups, to evaluate this scenario, we only look at the trends of the reputations distributions. \figurename~\ref{fig:genderAgeJobBWC}~(a) shows us that using solely Eq. \ref{eq:gramos} leads to a consistent reputation disparity, with reputation being distributed unequally across the groups. This phenomenon is clearly mitigated \figurename~\ref{fig:genderAgeJobBWC}~(b), where all the demographic groups have, on average, the same reputation.

\begin{figure}
\centering      
\subfigure[Reputation scores without reputation independence]{\includegraphics[width=0.4\textwidth]{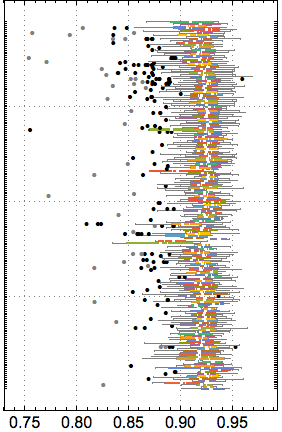}}\qquad
\subfigure[Reputation scores with reputation independence]{\includegraphics[width=0.4\textwidth]{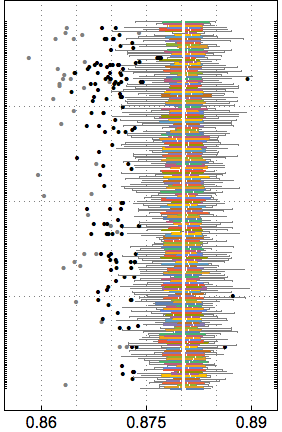}}
\caption{Box-whisker-chart for reputations of users resulting from Eq.~\eqref{eq:gramos} in (a) and (b), respectively, and from Eq.~\eqref{eq:gramos} plus Eq.~\eqref{eq:multi_fair} in (c) and (d), respectively, with $\lambda = 0.5$, for the multiple attributes \underline{gender}, \underline{age} and \underline{job} (ML-1M).}
\label{fig:genderAgeJobBWC} 
\end{figure}

\vspace{1mm} \noindent \textbf{BookCrossing}. 
We start by assessing the disparity originated by the original method in Eq.~\eqref{eq:gramos}. In a Box-whisker-chart (BWC) that considers age and location as attributes, Figure~\ref{fig:BC_age_loc}~(a) shows us that using solely Eq. \eqref{eq:gramos} leads to a consistent reputation disparity on age-based groups. Specifically, we can observe a pattern according to which, the younger are the users, the larger average reputation values the class has, thus yielding a bias on the attribute age. Table~\ref{tab:BC_1_DR} quantifies this disparity. To assess its significance, we test the null hypothesis that the mean difference between two classes of the attribute is 0, at the 5\% confidence level, under an LT test. The results in Table~\ref{tab:BC_1_DR} show that the disparity actually occurs only when the age gap between the users is large and only affects the groups of elder users, which are the less represented.

\begin{figure}
\centering      
\subfigure[reputation scores]{\includegraphics[width=0.4\textwidth]{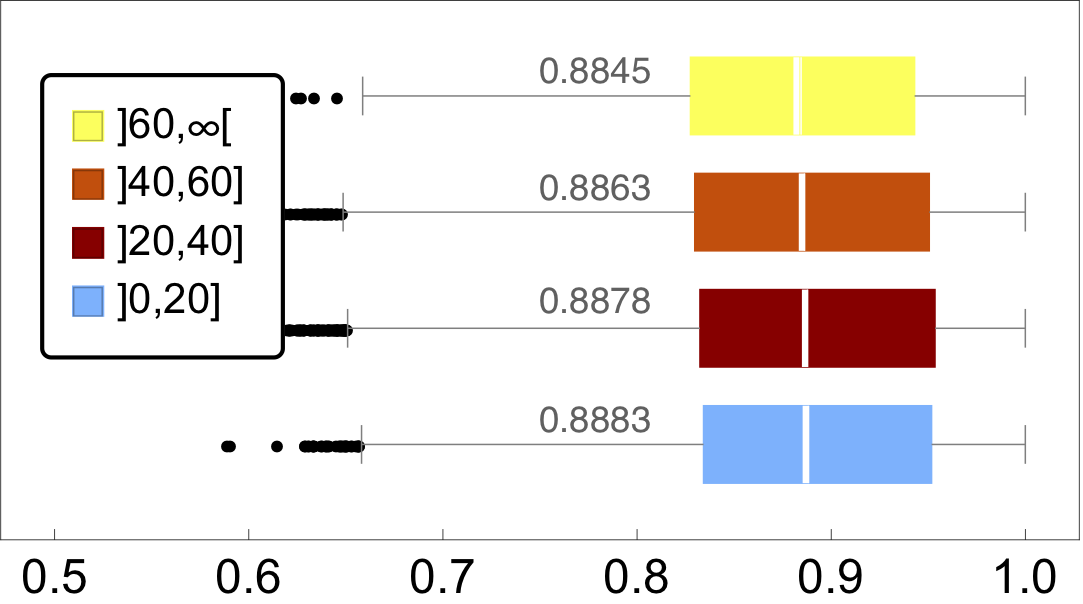}}\qquad
\subfigure[reputation scores]{\includegraphics[width=0.4\textwidth]{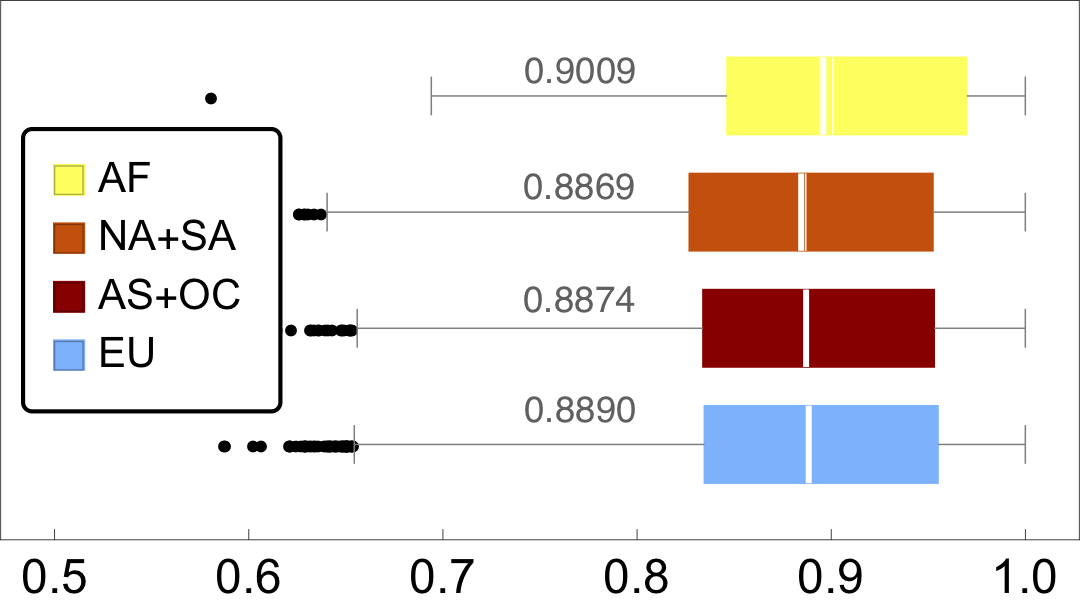}}
\\ 
\subfigure[reputation scores]{\includegraphics[width=0.4\textwidth]{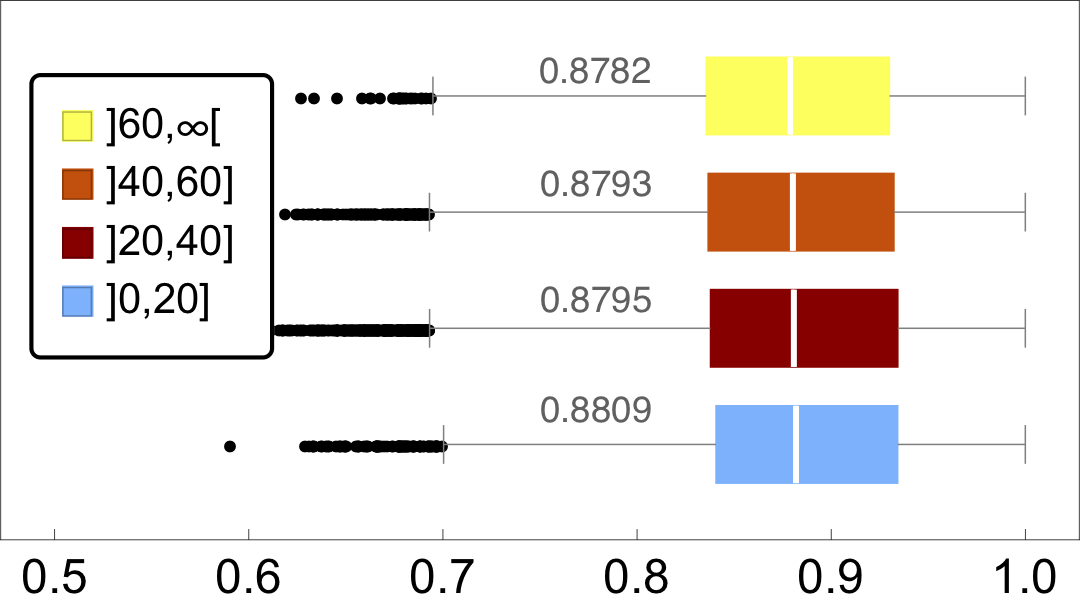}}\qquad
\subfigure[reputation scores]{\includegraphics[width=0.4\textwidth]{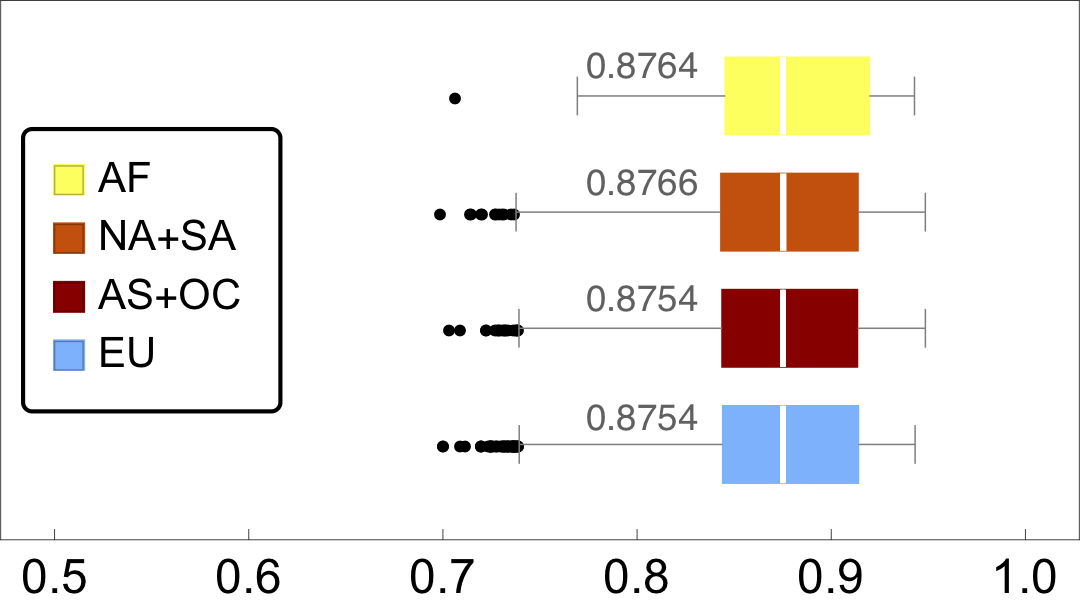}}
\vspace{-0.2cm}
\caption{Box-whisker-chart for reputations of users resulting from Eq.~\eqref{eq:gramos} in (a) and (b), respectively, and from Eq.~\eqref{eq:gramos} plus Eq.~\eqref{eq:multi_fair} in (c) and (d), respectively, with $\lambda = 0.5$ (where $\lambda$ should verify $\lambda\in]0,1[$) to give a medium penalization to the users with preferences that differ from the rest of the community, for the multiple attributes \underline{age} and \underline{location} (BookCrossing).} 
\label{fig:BC_age_loc} 
\end{figure}

\begin{table}
\centering
\begin{tabular}{|c|cccc|} \hline
 & $<20$ & $20-40$ & $40-60$ &  $>60$
\\ \hline
$<20$ &  -- & 0.0004 & \textbf{0.0020} & \textbf{0.0048} \\
$20-40$ & -- & -- & \textbf{0.0016} & \textbf{0.0044}   \\
$40-60$ & -- & -- & -- & \textbf{0.0028}  \\ \hline
\end{tabular}

\caption{DR of reputations resulting from Eq.~\eqref{eq:gramos}, for attribute age (BookCrossing). 
The LT tests for the reputations are denoted as normal text if $H_0$ is not rejected and bold text if $H_0$ is rejected.}
\label{tab:BC_1_DR}
\end{table}

\begin{table}
\centering
\begin{tabular}{|c|cccc|} \hline
 & EU & AS+OC & NA+SA & AF 
\\ \hline
EU & -- & 0.0014 & \textbf{0.0039} & \textbf{-0.0075}  \\
AS+OC  & -- & -- & \textbf{0.0025} & \textbf{-0.0089}  \\
NA+SA & -- & -- & -- & \textbf{-0.0114}  \\
\hline
  \end{tabular}

\caption{DR of reputations resulting from Eq.~\eqref{eq:gramos}, for attribute location (BookCrossing). 
The LT tests for the reputations are denoted as normal text if $H_0$ is not rejected and bold text if $H_0$ is rejected.}
\label{tab:BC_2_DR}
\end{table}

Moving to the location attribute, Figure~\ref{fig:BC_age_loc}~(b) shows the impact of Eq.~\eqref{eq:gramos} in a Box-whisker-chart (BWC). Results surprisingly indicate that, under this setting, on average, the smallest group (AF) obtains the highest reputation values. We conjecture that this might be because the group might represent a small and cohesive community. To deeply investigating this possible bias, Table~\ref{tab:BC_2_DR} quantifies the DR values and assesses if the mean difference between two classes of the attribute is 0, at the 5\% confidence level, under the LT test; a disparate reputation only occurs between European and American users, with the latter having a higher average reputation. This assessment of disparate reputation on the BookCrossing dataset leads us to our fourth observation.

\hlbox{Observation 4}{Under a multi-class attribute setting, considering a finer granularity when creating the classes facilitates the emergence of disparities (see the difference between the age attribute in the two datasets). Besides, multi-class attributes where two of the classes represent the vast majority of the user base (see location in BookCrossing) behave as binary attributes, leading to possibly uncovering disparities only between the two biggest classes.\hfill$\diamond$}

BookCrossing becomes an interesting benchmark to evaluate our approach, given that we are combining attributes where disparities do not occur for every combination of the classes. When introducing our multi-attribute reputation independence with Eq.~\eqref{eq:multi_fair}, Figure~\ref{fig:BC_age_loc}~(c) and 
the DR values show a disparity $\approx 0$. The LT tests
confirm that we cannot reject the null hypothesis. Thus, we can mitigate disparate reputation for attribute age. The same occurs for attribute location as observed in the BWC in Figure~\ref{fig:BC_age_loc}~(d), with the DR scores all $\approx 0$, 
and the LT tests 
confirm that we cannot reject the null hypothesis.

\subsection{Robustness Evaluation (RQ2)}\label{sec:ev_rob}
To answer the second question, we measured the system robustness with the Kendall Tau ($\tau$) metric, by comparing the rankings obtained without attackers and the rankings obtained after attacks~\citep{LiYHC12,saude2017robust}. Concretely, an attacker can be either a person or a bot interested in diminishing or increasing the rankings of a given item or set of items. ). Due to space constraints and because the results obtained with the two datasets are almost identical in terms of underlying patterns, we provide results only for ML-1M. Specifically, our study investigates the system's robustness before and after applying our mitigation method in Eq.\ref{eq:multi_fair}, when the following kinds of spamming/attacks are carried out:

\begin{itemize}
\item \emph{Random spam} consists of a set of users giving random ratings to a random set of items with fixed size;
\item \emph{Love/hate attack} consists of a set of users giving the highest rating to a target item and the lowest rating to a random set of items with fixed size;
\item \emph{Hate/love attack} consists of a set of users giving the lowest rating to a target item and the highest rating to a random set of items with fixed size.
\end{itemize}

For the sake of reproducibility, in all experiments, for love/hate and hate/love attacks, we chose item 1 and item 3 as the target items, respectively. For the three types of attacks, we selected the fixed size of the random set of items to be 10. For the attacked user, we randomly selected their attributes classes. 

First, we test the random spamming, by simulating a proportion of spammers ranging from $0.10$ to $0.35$ of the total number of ratings. 
The results in Figure~\ref{fig:exp_rob}~(a) show us that the robustness for mitigation methods in Eq. \ref{eq:gramos} and Eq. \ref{eq:multi_fair} is comparable, whereas both of them led to a slight improvement concerning the Arithmetic Average (AA). 

Second, we simulated two different attacks to the most voted item, ranging the proportion of attackers from $0.10$ to $0.35$ of the total number of voters of the target item. Figure~\ref{fig:exp_rob}~(b) depicts the results under the love/hate attack. Notice that, by mitigating disparate reputation, the attack is less effective for both methods. Indeed, robustness is significantly higher than the one obtained with the AA. In Figure~\ref{fig:exp_rob}~(c), the hate/love scenario leads to similar observations. 

\hlbox{Observartion 5}{In general, simultaneously mitigating biases on two sensitive attributes does not prevent the system's robustness, similarly to the single-attribute method. Furthermore, the method proposed in this paper increases the system's robustness when compared with the arithmetic average (AA).}

Regarding the \textit{BookCrossing} dataset, the results follow the same trend, and, therefore, the plots are omitted. 
In other words, the results of using Eq.~\eqref{eq:gramos} and the results of using Eq.~\eqref{eq:gramos} and Eq.~\eqref{eq:multi_fair} are essentially indistinguishable. 

\begin{figure}[!ht] 
\centering      
 \subfigure[Random Spam]{\label{fig:e}\includegraphics[width=0.295\textwidth]{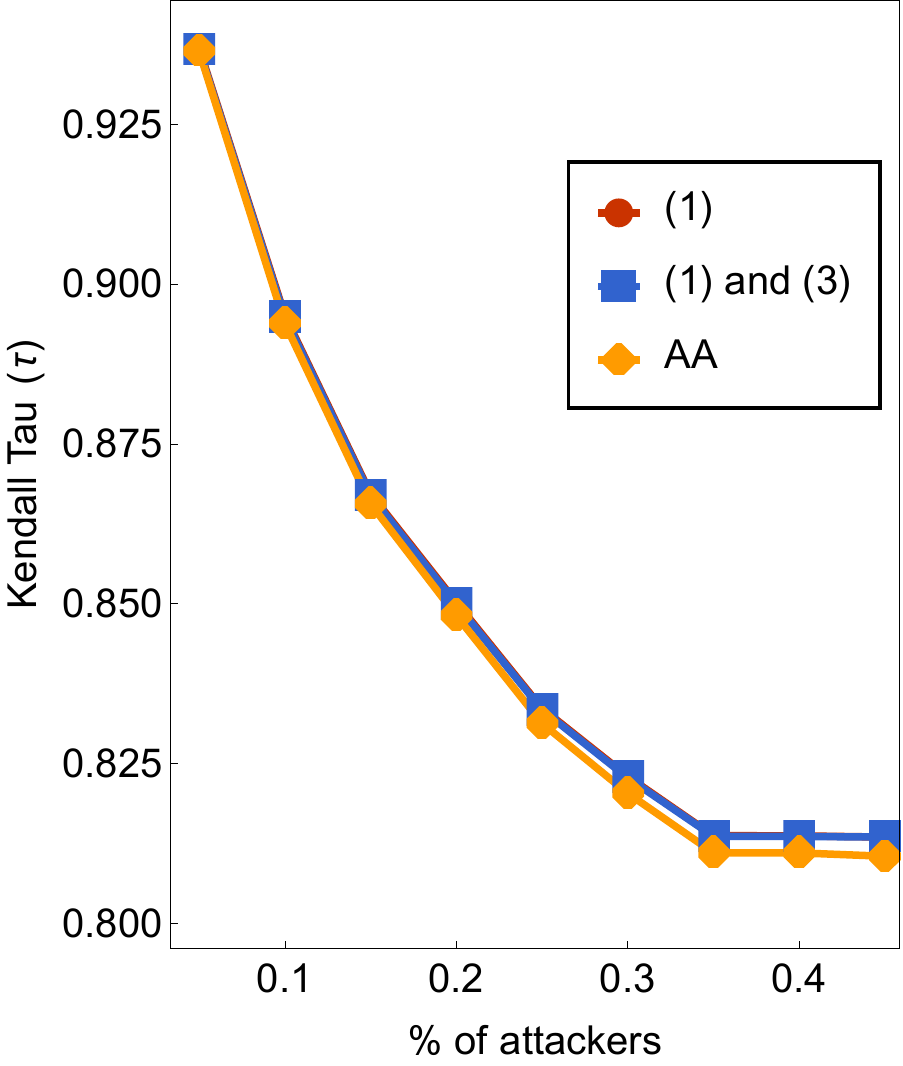}}\qquad
\subfigure[Love/hate]{\label{fig:f}\includegraphics[width=0.295\textwidth]{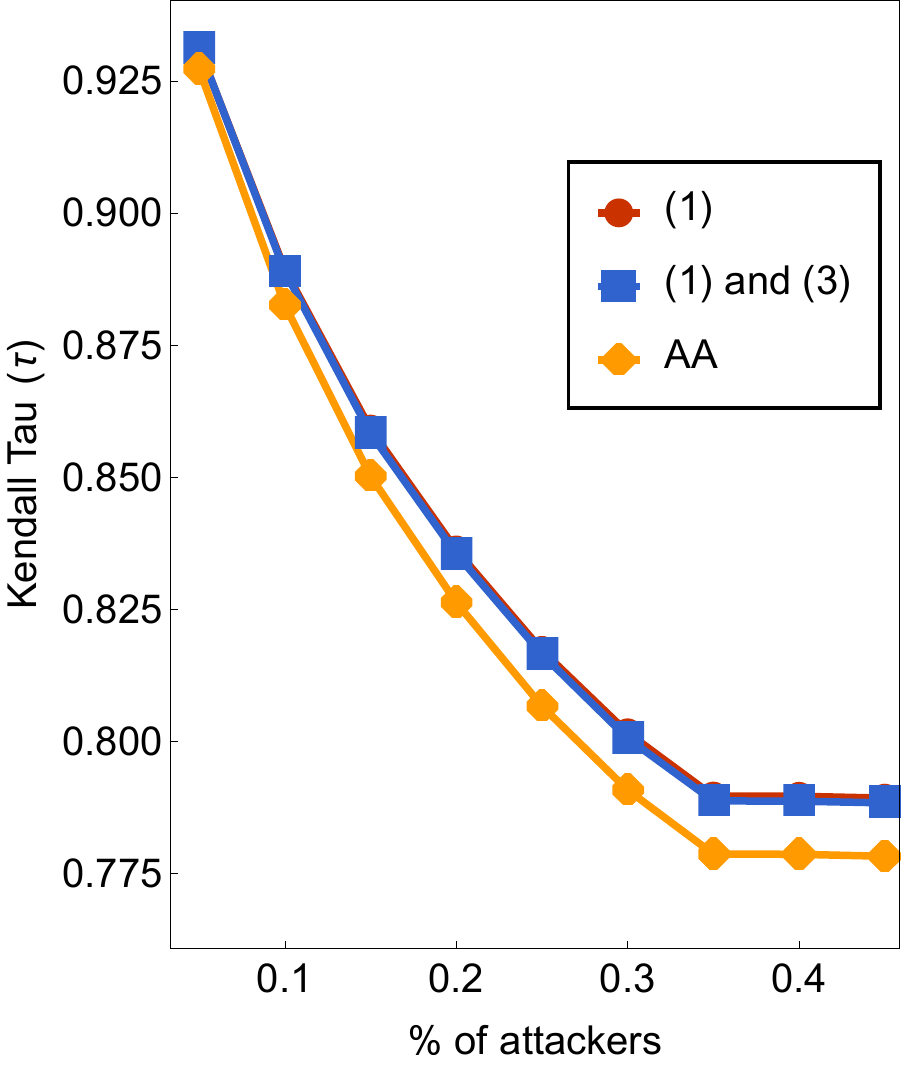}}\qquad
 \subfigure[Hate/love]{\label{fig:g}\includegraphics[width=0.295\textwidth]{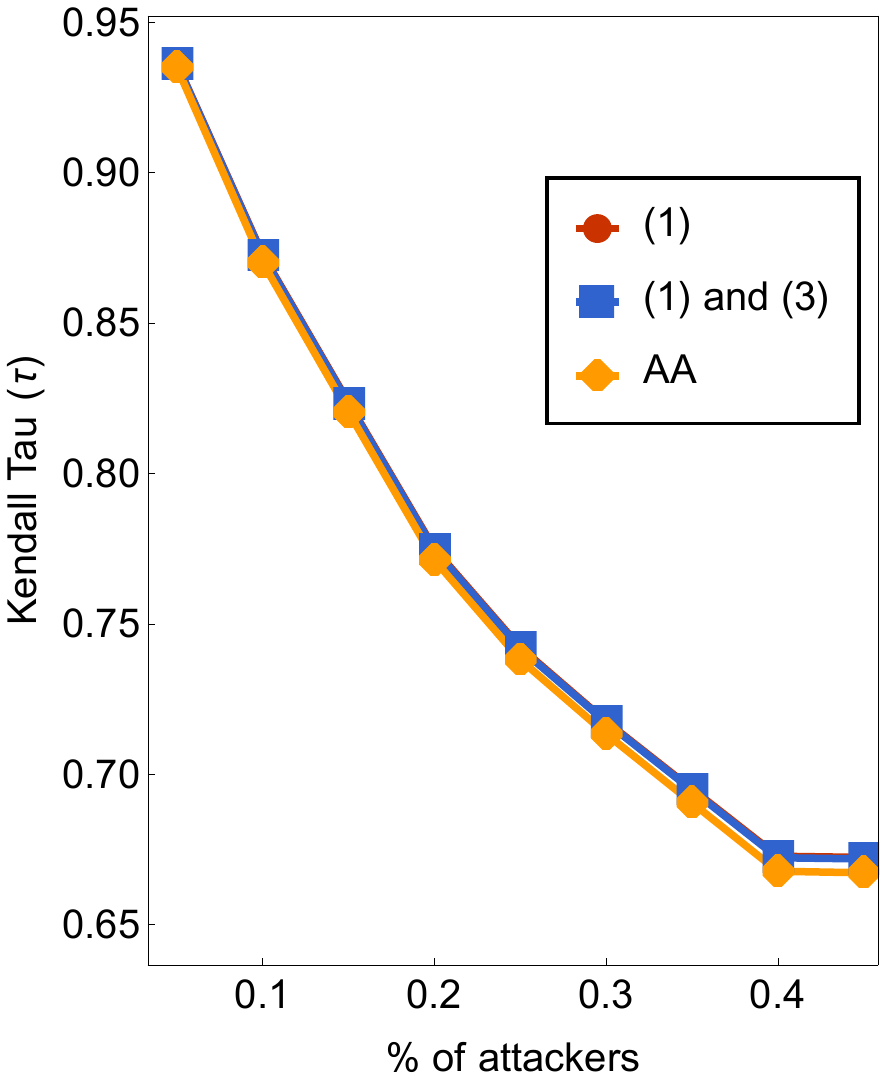}}
\caption{Robustness of the system before and after bias mitigation, with a percentage of attacked users relatively to the original number of users between 5\% and 45\%. Notice that the results of using Eq.~\eqref{eq:gramos} and the results of using Eq.~\eqref{eq:gramos} and Eq.~\eqref{eq:multi_fair} are essentially indistinguishable (ML-1M).}
\label{fig:exp_rob} 
\end{figure}

\begin{table}[t]
\centering
\subfigure[]{
\begin{tabular}{|c|ccccc|}
\hline
\multicolumn{6}{|c|}{\textsf{Dataset \textit{ML1M}}}\\ \hline
& $\emptyset$ & \textsf{Rep$_{\text{Gender}}$} & \textsf{Rep$_{\text{Age}}$} & \textsf{Rep$_{\text{Gender+Age}}$} & \textsf{Rep$_{\text{Gender+Age+Job}}$}
\\ 
$\tau$ & 0.9950 & 0.9954 & 0.9959 & 0.9961 & 0.9982\\
RMSE & 0.196207 & 0.196164 & 0.196164 & 0.194862 & 0.196171
\\
\hline
\end{tabular}
}
\qquad
\subfigure[]{
\begin{tabular}{|c|cccc|}
\hline
\multicolumn{5}{|c|}{\textsf{Dataset \textit{BookCrossing}}}\\ \hline
& $\emptyset$ & \textsf{Rep$_{\text{Age}}$} & \textsf{Rep$_{\text{Loc}}$} & \textsf{Rep$_{\text{Age+Loc}}$}
\\ 
$\tau$ & 0.9914 & 0.9914 & 0.9913 & 0.9914\\
RMSE & 0.275991 & 0.275215 & 0.275104 & 0.274718
\\
\hline
\end{tabular}
}
\caption{Quality of the obtained rankings for both \textit{ML-1M} and \textit{BookCrossing} datasets, respectively (a) and (b). \textsf{Rep} denotes the intervention on the reputation scores and its subscript indicates the considered sensitive attribute (\textsf{Gender} for gender, \textsf{Age} for age, \textsf{Loc} for location, \textsf{Job} for job). } 
\label{tab:eff}
\end{table} 

\subsection{Ranking Quality (RQ3)}

Finally, to evaluate the impact of the proposed method on ranking quality, we use the Kendall Tau~\citep{kendall1938new} with AA as the ground truth, as done in~\cite{LiYHC12}. Moreover, to assess to what extent the produced rankings reflect the individual ratings, we compute the RMSE, by splitting the data into training and test, where 90\% of the ratings was used to shape the rankings and the remaining 10\% for testing.
We report the observed $\tau$ and RMSE for each of the attributes considered in Table~\ref{tab:eff}.  When observing the $\tau$ values, we notice that, generally, the Kendall Tau   slightly improves when we mitigate bias relative to a grouping.  Moreover, when observing the {\em ML1M} dataset, when we consider more attributes (i.e., smaller groups), the values of $\tau$ increase. This improvement means that our approach yields a rankings' order closer to AA, yet assigning different relevance  (reputation) to users, w.r.t.~(\ref{eq:gramos}).  When we consider the $\tau$ values returned by {\em BookCrossing} dataset, we can observe that, when considering different groupings of the users, the values remain constant with respect to the state-of-the-art method we compare with (i.e., the one from~\cite{saude2017robust} that, in turn outperforms~\cite{li2012robust} in terms of robustness). 
This is a sign that, {\em when providing reputation independence, ranking quality is not affected}. This confirms the positive impact that our approach can provide to both the users and the platform, since the platform does not have to compromise ranking quality to provide equity in terms of reputation and robustness in case of attacks. This can generate trust on the users, since the rankings they interact with would (i) reflect their preferences, (ii) be unbiased on their sensitive attributes, (iii) be robust to malicious ratings. 

Considering RMSE, we can observe that {\em ML1M} returns a very low error if we consider the original rating scale, which was in the range [1,5]. {\em BookCrossing} show a slightly higher error, which we conjecture might be due to the fact that the dataset is larger; hence, providing a unique item ranking that reflects the  preferences of all the users is a more challenging task. Also for the RMSE, we can observe that, regardless of the dataset, the values remain constant, thus confirming the capability of our approach to enable reputation independence and robustness, while keeping ranking quality stable.

 In conclusion, the reputation concept treats users differently, which leads to a ranking with a bias for specific users' attributes. 
With the proposed approach, for a specific attribute, we mitigate bias. 
With our method, the concept of reputation still plays a role inside each group with a particular attribute value, but it does not cause bias. 
So, we get ``closer'' to the average as AA does not treat groups differently. Furthermore, with our approach, we also do not treat demographic groups differently. 

This leads us to our final observation, which connects the results from the three perspectives we analyzed (namely, reputation independence, robustness, and ranking effectiveness.

\hlbox{Observation 6}{The introduction of reputation independence allows to produce unbiased rankings w.r.t. to the sensitive attributes of the users, thus reflecting better their preferences.  This result can be achieved without affecting neither the platform in terms of the robustness it offers to the users (the values of $\tau$ remain unaltered), nor the users, since ranking quality in terms of RMSE remains the same. }

\subsection{Discussion}
Reputation-based ranking systems aim to rank the items by ensuring the community's preferences as a whole are reflected in the way items are sorted. It thus becomes essential to compute less biased formulations of user reputation, to weigh individual preferences without deterring other system properties~\citep{SaDAJSK18}. 

Past work in reputation-based ranking systems showed that approaches that compute reputation scores often make heavily biased decisions~\citep{RamosB20,MayWH19}. These results spurred investigations on the revision of reputation-based ranking algorithms and result in classic systems to mitigate such biased on even unfair decisions. Our theoretical formulation connects to previous studies in reputation-based ranking systems. Under this view, we believe that our empirical results suggest that a multi-attribute mitigation method can preserve the essential properties of a non-personalized ranking system.  
Moreover, these aspects are not guaranteed simultaneously by any other state-of-the-art method applied to the same class of algorithms. Our extensive experiments provide evidence that our method results in less biased reputation scores and can lead to more robust systems against attacks.  

The results we obtained with ML-1M have highlighted that gender is a central source of bias, leading to the highest disparate reputation estimates. When gender is not available, as in  BookCrossing, capturing disparate reputation phenomena via other characterizing attributes, such as age or geographic provenience, becomes more challenging. While our results have shown that the approach proposed in this paper can effectively mitigate disparities, regardless of the attributes it treats, the assessment of such phenomena can be challenging when possible sources of disparity are deeply hidden in the data.

\subsection{Limitations}

Our mitigation method is flexible to incorporate more elaborate conditions with more than two attributes. However, our study also embraces some possible limitations presented in what follows. 

\vspace{1mm} \noindent \textbf{Data-related limitations}. The datasets we considered in this study offered different sensitive attributes that characterize users or do not allow to arrange the same classes for a given attribute to have statistically valid results. It follows that we could not coherently compare the impact on a given attribute in a different domain (we analyzed the case of gender in our results' discussion). If new datasets offering the same set of sensitive attributes become available, it will be interesting to analyze this perspective. Furthermore, not having more than one dataset offering both ratings and demographic attributes in a single domain (e.g., movies or books) also does not allow us to provide a characterization of how demographic attributes impact on the reputation of the users in that domain. Again, the appearance of new datasets would allow deepening the inspection of such phenomena.

\vspace{1mm} \noindent \textbf{Evaluation-related limitations}. First, we evaluated our approach with the Disparate Reputation metric, and with the statistical two-sample location test (LT). The LT test requires the population of each class to have a significant number of users and, therefore, might be limiting in scenarios with nice user groups. 
The DR does not suffer from that problem, but it only uses average reputations. Therefore, as future work, we may design other evaluation metrics that may better capture the concept of reputation bias. 
Second, assessing the effectiveness of our approach in terms of accuracy of the ranking is a challenging aspect due to the lack of ground truth of what a good non-personalized ranking is. 
This fact connects to the lack of evaluation metrics to assess the accuracy for groups of users, which is also an open issue highlighted in the group recommendation research area~\citep{Boratto16}.
Third, plugging our solution into a specific reputation-based ranking system means that we could not assess the impact of our solutions on different types of ranking. As we stated in Observation~2 (Section~\ref{sec:mit_mul}), our solution can be embedded in {\em any} ranking system that computes rankings as a weighted average of the ratings. However, evaluating our solution on other reputation-based ranking systems (as those in Sec.~\ref{sec:related}) is left as future work. Lastly, we assessed the ranking system as a whole to provide actionable insights to service providers. However, a more fine-grained analysis would allow us to understand how robust is the system when bribing different demographic groups (e.g., assessing if minorities are also more vulnerable). To focus on our core contribution, we leave this perspective as future work.

\vspace{1mm} Despite these limitations, our multi-attribute mitigation method opens to new avenues of research in the field of reputation-based ranking systems, with a clear connection to other retrieval systems.

\section{Impact on the Web Society}\label{sec:impact}
Inequalities based on gender or ethnicity are present in both the online and offline world~\citep{WachsHVD17,HannakWGMSW17,Thebault-Spieker15,GeKMZ16}. Focusing on the concept of reputation, it has been observed that, if in social offline situations the identity of a person is disclosed, the reputation of females is lower than that of males~\citep{jones2012gender}. The previously mentioned study has shown that gender biases in reputation can be removed by hiding the identity of a person, with females being considered as valuable as their counterparts. Our study shows that this is not the case with ranking algorithms, which can learn biased patterns even though the algorithm is not fed with sensitive attributes of the users. 

Our class of ranking systems considers automatically computed notions of user reputation. 
Platforms like StackOverflow associate explicit reputation scores to the users. Also in this case, for female users, a lower reputation can be observed~\citep{MayWH19}. While the authors have tried to explain this gap in reputation between different genders (e.g., by considering the difference in participation in the platform between males and females), 11\% of the reputation gap remains unexplained. Hence, regulating disparate reputation via our approach can have a positive impact on online platforms that work with explicit reputation formulations, thus avoiding the discrimination and biased representation of legally-protected groups. This is of paramount importance when disparities cannot be explained and mitigated, as in~\cite{MayWH19}.

Another phenomenon observed in the literature is that most of the studies consider only a single source of bias (e.g., gender or ethnicity), e.g., by showing that Wikipedia~\citep{WagnerGGM16} and both OpenStreetMaps and Google MapMaker~\citep{Stephens13} provide gender-biased representations of knowledge. However, humans are such complex beings, that even if we knew all the attributes that characterize a person, it would impossible for an algorithm to understand and support us in our entirety, considering all the nuances that make us who we are. Nevertheless, trying to cover as many perspectives as possible, as our algorithm tries to do, is better than considering a single perspective (e.g., gender). For this reason, we believe our approach is a good first step towards assessing and mitigating multiple forms of bias, associated with sensitive attributes of the users, to protect more vulnerable groups and minorities. Finally, our paper moves a step forward towards shaping a blueprint of the decisions and processes to be done, when multiple sensitive attributes need to be considered in reputation-based ranking systems.

\section{Conclusions}\label{sec:conc}
Mitigating bias in reputation-based ranking systems is of paramount importance to ensure that the whole community's preferences are reflected in the way items are ranked, without being biased against users' sensitive attributes.
Our study in this paper analyzed if mitigating reputation bias for two sensitive attributes individually (for any order) yields less biased reputations for both attributes. Our results uncovered that existing countermeasures do not guarantee this critical property. Based on this finding, we proposed a novel approach aiming to ensure reputation independence for multiple sensitive attributes simultaneously. Our experiments on real-world data showed that our mitigation can achieve the proposed goal. Moreover, we accomplished the envisaged goal without hindering essential system's qualities, e.g., ranking quality and robustness against attacks. 

Future work will embrace the findings and limitations of this study to drive research on unexplored ranking domains and disparate reputation measures, with positive impacts on social good. Specifically, we plan to conduct analyses of the impact on robustness to different demographic groups, and to design group-based metrics to assess the effectiveness of ranking systems.

\section*{Declarations}

\noindent\textbf{Funding} This research was partially supported by by the Portuguese Funda\c c\~{a}o para a Ci\^{e}ncia e a Tecnologia (FCT) through the FCT project RELIABLE, Portugal (PTDC/EEI-AUT/3522/2020), and by ACCIÓ, under project ``Privacy-preserving, Fair and Explainable Artificial Intelligence (PrEFair)''.
\\

\noindent\textbf{Conflicts of interest/Competing interests}
The authors confirm there are no conflicts of interest.
\\

\noindent\textbf{Availability of data and material} We used only datasets that a publicly available.
\\

\noindent\textbf{Code availability}
The code is published on \url{https://fenix.tecnico.ulisboa.pt/downloadFile/1407993358910329/Fair_Reputation_Based_Ranking.nb}\\
\\

\noindent\textbf{Ethics approval} The work uses publicly available and non-identifiable information of the users. No ethical approval was needed.
\\

\noindent\textbf{Consent to participate} Not applicable, since no human participant was involved in the evaluation of our study.
\\

\noindent\textbf{Consent for publication} Not applicable, since all datasets used in this study are released by third parties.
\\

\noindent\textbf{Authors' contributions}
G. Ramos contributed to the design of the solution, design of the algorithm, the coding of the solution, the evaluation of the approach, and the writing of the paper. 
L. Boratto contributed to the design of the solution, the analysis of the results, and the writing of the paper.
M. Marras contributed to the design of the solution, the analysis of the results, and the writing of the paper.

\bibliographystyle{spbasic_updated}      
\bibliography{sample_library}

\end{document}